\documentclass[a4paper,10pt,twoside
]{article}
\pdfoutput=1

\usepackage{geometry}
\usepackage[utf8]{inputenc}
\usepackage[T1]{fontenc}
\usepackage[english]{babel}

\usepackage{cite}
\usepackage{url}

\usepackage{xpatch}
\usepackage{amsmath}
\usepackage{amsthm}
\usepackage{amssymb}
\usepackage{dsfont}             
\usepackage[caption=false,font=footnotesize] 
{subfig}

\usepackage{graphicx}

\usepackage{mathtools}
\mathtoolsset{centercolon=true,%
              mathic=true
}
\usepackage{algorithm}                  
\usepackage{algpseudocode}              

\usepackage{hyperref}

\makeatletter
\xpatchcmd{\algorithmic}{\itemsep\z@}{\itemsep=0.15ex}{}{}
\makeatother

\newcommand{\lra}{\xRightarrow[\phantom{bla}]{}}

\newcommand{\C}{\mathbb{C}}
\newcommand{\N}{\mathbb{N}}
\newcommand{\lN}[1]{$\ell_{#1}$}

\newcommand{\muW}{\mu_\text{Welch}}

\newcommand{\snr}{\text{SNR}}

\newcommand{\ti}[1]{\tilde #1}

\newcommand{\II}{\mathds{1}}

\newcommand{\eOMP}{\hbox{$\epsilon$-OMP}}

\newcommand{\Fr}{\textit{F\_rand}}
\newcommand{\Fc}{\textit{F\_consecutive}}
\newcommand{\Fcb}{\textit{F\_consecBegin}}
\newcommand{\Fx}{\textit{F\_nXStatBlocks}}
\newcommand{\Rg}{\textit{R\_gauss}}

\newcommand{\figref}[1]{\hbox{Fig.~\ref{#1}}}

\newcommand{\abs}[1]{\left|#1\right|}
\DeclareMathOperator*{\argmax}{argmax}

\newcommand{\bracks}[1]{\left[#1\right]}
\newcommand{\ceil}[1]{\left\lceil#1\right\rceil}
\newcommand{\clos}[2]{\text{clos}_{#1}\left(#2\right)}

\newcommand{\floor}[1]{\left\lfloor#1\right\rfloor}
\newcommand{\norm}[1]{\left\lVert#1\right\rVert}
\newcommand{\paren}[1]{\left(#1\right)}
\newcommand{\setbr}[1]{\left\{#1\right\}}
\newcommand{\scalarP}[2]{\left<#1,#2\right>}

\newcommand{\supp}[1]{\text{supp}\left(#1\right)}

\theoremstyle{plain}

\newtheorem{thm}{Theorem}

\newtheoremstyle{remarkOwn}
{\topsep}
{\topsep}
{\upshape}
{0pt}
{\bfseries}
{}
{0.5em}
{}
\theoremstyle{remarkOwn}

\newtheorem{concl}[thm]{Corollary}

\newtheoremstyle{definitionOwn}
{\topsep}
{\topsep}
{\upshape}
{0pt}
{\bfseries}
{}
{.5em}
{}
\theoremstyle{definitionOwn}
\newtheorem{df}[thm]{Definition}

\author{Tobias Birnbaum%
\thanks{Tobias Birnbaum, Department of Electronics and Information Processing (ETRO), Vrije Universiteit Brussel, 1050 Brussels, Belgium, E-mail: 
\hbox{Tobias.Birnbaum@vub.ac.be}.}%
, Yonina C. Eldar (IEEE \textit{Fellow})%
\thanks{Yonina C. Eldar, Department of Electrical Engineering, Technion-Israel Institute of Technology, Haifa, Israel, E-mail: 
Yonina@ee.technion.ac.il.}%
 ,  Deanna Needell (IEEE \textit{Member})%
\thanks{Deanna Needell, Department of Mathematical Sciences, Claremont McKenna College, CA 91711, USA, E-mail: DNeedell@cmc.edu.}%
}
\date{\today}
\title{Tolerant Compressed Sensing With Partially Coherent Sensing Matrices}
\providecommand{\keywords}[1]{\textbf{\textit{Index terms---}} #1}

\begin{document}

\setcounter{secnumdepth}{4}
\maketitle

\begin{abstract}
Most of compressed sensing (CS) theory to date is focused on incoherent sensing, that is, columns from the 
sensing matrix are highly uncorrelated. However, sensing systems with naturally occurring correlations arise in many applications, such as signal detection, motion 
detection and radar. 
Moreover, in these applications it is often not necessary to know the support of the signal exactly, but instead small errors in the support and signal are tolerable. 
Despite the abundance of work utilizing incoherent sensing matrices, for this type of \textit{tolerant} recovery we suggest that coherence is actually 
\textit{beneficial}. We promote the use of coherent sampling when tolerant support recovery is acceptable, and demonstrate its advantages empirically. In 
addition, we provide a first step towards theoretical analysis by considering a specific reconstruction method for selected signal classes.
\end{abstract}


\section{Introduction}\label{sec:introduction}
Compressed sensing (CS) deals with sampling and recovery of sparse signals\cite{Yonina2012Book, Foucart2013Book, Eldar2015Book}. By using the sparsity structure, 
recovery is possible from far fewer measurements than the signal length. Initial results (e.g. \cite{Donoho1989Uncertainty,Candes2006NearOptimal}) showed that it is 
possible to approximate the NP-hard \lN{0} minimization with optimization problems that have only polynomial complexity, such as \lN{1} minimization or orthogonal 
matching pursuit (OMP) \cite{Tropp2004}.

Although most classical results are in terms of $\ell_2$ or exact support error, we focus here on the notion of \textit{\hbox{$d$-tol}\-erant recovery}, motivated 
by applications such as geophysics and radar\cite{BarIlan2014, 
Eldar2015Book}. By \hbox{$d$-tol}\-erant recovery, we mean signal (or support) recovery in which one tolerates errors in signal spike locations of up to $d$ indices. In 
other words, the position of every non-zero in the reconstructed support can differ up to $d$ indices from the original support. In these applications just mentioned, 
for example, since the scene is typically discretized along a fine grid, one often does not need precise target/event location but rather can tolerate a small amount of 
spatial error.

We demonstrate that we can increase noise robustness by using \hbox{$d$-tol}\-erant recovery and special types of \textit{partially coherent matrices}. 
This is in contrast to the majority of results in CS where incoherent sensing matrices are 
highly desirable - e.g. \cite{Yonina2012Book,Foucart2013Book,Tropp2003Improved,Candes2006NearOptimal,Donoho2006stable,Candes2007Sparsity}. The use of partially coherent 
sensing matrices provides a new avenue to pursue for applications where such matrices arise naturally and/or where small errors are acceptable. 
Typical applications are reconstructions of multi-band signals \cite{Mishali2009},
unions of subspaces \cite{Eldar2009},
signal/image processing such as super-resolution \cite{Chen2014} or face recognition algorithms \cite{Wright2009}.


{\bfseries Contribution.}
Our goal is to introduce the notion of \hbox{$d$-tol}\-erant recovery and demonstrate that partially coherent matrices are beneficial in this context. We view our main 
contribution as two-fold: (i) we demonstrate that if an application requires the use of coherent sampling, then \hbox{$d$-tol}\-erant recovery \textit{is still 
possible}, and moreover (ii) that if the desired outcome is actually a tolerant recovery, then one actually \textit{should} use coherent sampling.  To our best 
knowledge, these phenomena have not been adequately observed, explored or studied, except for preliminary work in the thesis of Bar-Ilan \cite{BarIlan2014Thesis}, which 
is the motivation of our work here. We demonstrate these ideas through empirical results and also establish a foundation for theoretical guarantees under specific 
(non-optimal) assumptions.

{\bfseries Organization.}
The structure of the paper is as follows: In Section~\ref{sec:motivation} we motivate \hbox{$d$-tol}\-erant recovery and point out links to related work. 
Section~\ref{sec:concepts_and_definitions} provides a problem formulation and definitions necessary to capture \hbox{$d$-tol}\-erant theory. We present numerical 
simulation results comparing incoherent and partially coherent sensing matrices in Section~\ref{sec:numerical_simulation_results}. In Section~\ref{sec:theorems} we 
provide initial analytical justification for our observations under the assumption of sufficiently spread signal support using a variant of OMP \cite{Tropp2004}. 
 The work is concluded with a summary and 
outlook in Section~\ref{sec:conclusion_and_future_directions}.

{\bfseries Notation.}
For a positive integer $N$ we write $[N]$ to denote the set $\{1, 2, \ldots, N\}$.
The norms $\norm{\cdot}_p, p\in[1,\infty]$ refer to the vector norms in \lN{p} or the induced matrix norms. The number of non-zeros of a vector is denoted as 
$\abs{\cdot}_0$. Lower case Greek letters name the columns of the respective matrix. The $N$th order Fourier matrix is denoted as $F_N$. An $S$-sparse signal $x\in\C^N$ 
has exactly $S$ non-zeros. The reconstruction of $x$ from linear measurements $y\in\C^M$ is termed $\ti{x}\in\C^N$. We set $\Sigma:=\supp{x}$, $\Gamma:=\supp{\ti{x}}$ 
and always have $0<S\leq M \ll N$.

\section{Motivation}\label{sec:motivation}
\subsection{d-tolerant recovery}\label{sec:sub_d_tolerant_recovery}
We consider a \hbox{$d$-tol}\-erant recovery of an unknown signal $x$ from measurements $y$ given by the linear sensing model 
\begin{equation}
  y = \Phi x + e, \label{eq:linear_sensing_model_w_noise}
\end{equation}
with sensing matrix $\Phi\in\C^{M\times N}$ and measurement noise \hbox{$e\in\C^M$}. We assume that the vector $x$ is $S$-sparse, namely, $\abs{x}_0=S$. 

We postpone until Section \ref{sec:concepts_and_definitions} a formal definition of this tolerance, but informally we mean recovery which tolerates errors in the support set of up to $d$ indices. This aligns with applications in which the signal spikes refer to e.g. spatial locations, and one tolerates identified locations within $d$ units of actual locations.
Specifically, we seek a \hbox{$d$-tol}\-erant recovery of $x$ with $0<d$, $0<S\leq M \ll N$. For simplicity and to preserve the clarity of illustration we focus on 
\hbox{$d$-tol}\-erant recovery for the well known example of Fourier sensing matrices, although extensions to other settings are straightforward. Below, we construct 
several sensing matrices and investigate their performance in \hbox{$d$-tol}\-erant recovery.

Before proposing our coherent sampling approach and showing our results, we first mention some simple alternatives to tolerant recovery, along with their models.  We 
will use these models for testing purposes in later sections. Note that \hbox{$d$-tol}\-erant recovery aims to recover spike locations up to a spatial tolerance of $d$ 
indices.  A related but simpler viewpoint would group the coefficients of the signal into bins, each of size $d$, and hope to identify which bins contain spikes.  
Therefore, the most basic model for tolerant recovery would be to use an appropriate subsampled sensing matrix. This can be done in several ways, which we outline here. 
In all cases we aim for a measurement vector $y\in\C^M$. To be concrete, to downsample a vector $x\in\C^N$ to one in $\C^M$, we apply a downsampling matrix whose rows 
consist of single blocks of $N/M$ $1$s (and the rest zero). To upsample we simply pad the signal with zeros such that each entry of a downsampled block is mapped to 
the center of that block. We refer to these operations by $D$ and $U$, respectively. We denote by $e$ a noise vector of appropriate dimension.
\begin{itemize}
\item {\bfseries Subsampling on coarse grid:} Consider an $ \lceil N/d\rceil\times  \lceil N/d\rceil$ DFT matrix $F_{N/d}$.  Create the subsampled matrix obtained 
from $F_{N/d}$ by subsampling $M$ rows (as in any fashion described above). Reconstruct a vector $x$ of length $\lceil N/d\rceil$ using a classical CS reconstruction 
method.
\end{itemize}
For naive comparisons, we also consider two other scenarios. 
\begin{itemize}
\item {\bfseries Downsample then sense (DS):} In this case we consider first downsampling the signal $x\in\C^N$ to obtain a signal ${x_M}\in\C^M$.  Then we measure $y = 
F_M{x_M} + e$. To reconstruct, we simply apply $F_M^{-1}$ to the measurements $y$ and then upsample the result to obtain a reconstruction of $x$, $\hat{x} = 
U(F_M^{-1}y)$.
\item {\bfseries Sense then downsample (SD):} Here we first apply $F_N$ to the signal $x\in\C^N$ and then downsample the result to obtain $y = D(F_N x) + e$. To reconstrsuct we first upsample the measurements and then apply the inverse: $\hat{x} = F_N^{-1}(Uy)$.
\end{itemize}

We will see later that in most cases, when tolerant recovery is the goal, coherent sampling with our approach outperforms these simple methods. Of course, in other cases, the application may necessitate the need for coherent sampling, in which case our results show that tolerant recovery is still possible. Before formulating the details of tolerant recovery, we first review some related work.

\subsection{Related work}\label{sec:sub_related_work}
Partially coherent sensing matrices have been studied previously in CS. However, existing work has focused on \textit{exact} support recovery \textit{despite} coherence 
within the sensing matrix. Here, instead, we show that coherence is actually a \textit{resource} when we allow for \hbox{$d$-tol}\-erant recovery.

The literature on OMP related methods using partially coherent sensing matrices can be summarized as follows. 
In \cite{Fannjiang2012} multiple extensions to existing algorithms were formulated. The authors proved and showed numerically that by introducing a band-exclusion 
method they were able to recover signals in a specific sense.  Each non-zero of the original signal has a counterpart in the reconstruction, which is however 
allowed to be located \textit{anywhere}. Thus the "tolerance" would be $d=N-1$. 
Further, a condition related to the ERC\cite{Tropp2004} is required, and the signals are assumed to have support which is spread enough so that coherent columns do not appear in the support indices.
The work \cite{duaBar2013} also considers spread signals, seeking accurate signal recovery and attempting to overcome coherence in the sampling matrix.

In \cite{Eldar2010}, useful concepts such as the distinction between block coherence and sub coherence were developed and applied to 
the recovery of block-sparse signals using the block OMP (BOMP) algorithm.
Correlations were allowed across blocks, but each block itself must be incoherent. The results were refined in a generic manner yielding a block RIP 
in \cite{Wang2011}. The work in \cite{Fu2014} extended this framework to noiseless recovery from partially coherent sensing matrices with a static predefined 
column-block structure, using a block RIP as a necessary requirement. This was done still with the focus on accurate recovery of block sparse signals when the block structure is known a priori.

Along a different line of work, \cite{CP2009} shows that mild coherence in the sensing matrix can be allowed when the signal is modeled as random. In this case, accurate recovery is still possible when the coherence scales like $1/\log(N)$. Here again, in this setting the goal is exact recovery and the coherence is something that needs to be overcome, not something that aids in recovery.

Some results on exact recovery with dictionary sparsity models ($y=\Phi D x$) were derived in \cite{Candes2011Compressed,Giryes2013}.  
The proposed \hbox{D-RIP} condition was defined for the \hbox{\lN{1}-analysis} problem. This condition allows for coherence within the dictionary, $D$, but only $Dx$ is 
the target of the reconstruction; the sensing matrix $\Phi$ is still required to be incoherent. The same is true for the \hbox{\lN{1}-synthesis} problem which 
was treated in \cite{epsOMP} via \eOMP. The presented theoretical results are based on the $\epsilon$-coherence between the sensing matrix $\Phi$ and the partially 
coherent dictionary $D$. A recent surge of work has studied the area of dictionary sparsity models \cite{Giryes2013, Rubinstein2013, NOMP, Chen2014, 
Foucart2016}, all still requiring incoherence of the sensing matrix.

Related to these results but fundamentally different, is the super-resolution problem. In this problem, one only has information about a signal in its low frequency band, and wishes to obtain a higher resolution reconstruction from that data. This can be modeled as a CS problem where the sensing matrix is highly coherent and the signal has a spread out support. Recent work on this problem has shown that several optimization based or greedy methods are successful in accurately recovering these types of signals \cite{Granda2012,Demanet2013,Chen2014}.  Although later we will also consider spread signals, these works are fundamentally different than ours since their goal is exact reconstruction that \textit{overcomes} the coherent sensing, whereas we are promoting the \textit{advantages} of coherence sampling when tolerant detection is the goal.

To our best knowledge, the first observation that coherence in the sensing matrix is not only tolerated but even \textit{beneficial} for tolerant recovery appeared in the thesis of Bar-Ilan \cite{BarIlan2014Thesis}.

\section{Problem formulation and definitions}\label{sec:concepts_and_definitions}
In general, a \hbox{$d$-tol}\-erant recovery will be called successful if every non-zero of the $S$-sparse signal $x$ has a non-zero within the recovery $\ti{x}$ that 
is not further than $d$ indices apart. The success can be measured by the (relative) \hbox{$d$-tol}\-erant support recovery error. We define the $d$-closure of a column 
index~$i$ as
\begin{equation}
  \clos{d}{\setbr{i}}:=\setbr{\max\setbr{i-d,1},\dots,\min\setbr{i+d,N}}\;.
\end{equation}
The (relative) \hbox{$d$-tol}\-erant support recovery error measure is defined as
\begin{align}
  \rho_d\paren{\ti{x}, x} &:= \frac{\sum_{i\in\Sigma} \II\setbr{ \paren{\sum_{j\in\clos{d}{\Gamma}} \delta_{i,j} } > 0}}{S}
    \label{eq:def_dtol_recovery} \;,
\end{align}
with the indicator function $\II$, $S:=\abs{x}_0$, the Kronecker delta $\delta_{i,j}$, $d$-closure of the set $\Gamma$, $\clos{d}{\Gamma}:=\bigcup_{i\in\Gamma} 
\clos{d}{\setbr{i}}$, and other notation defined in the notation section above.

For block sparse signals, which have their non-zeros cumulated in blocks, this usually means that multiple non-zeros are combined to form a single representative for at 
most $(2d+1)$ non-zeros of a block.

The maximal number of non-zeros that can be resolved in a \hbox{$d$-tol}\-erant recovery within a signal of length $N$ is given as:
\begin{equation}
  S_{\text{max}}=\floor{\frac{N-1}{2d+1}}+1\;. \label{eq:maximal_s}
\end{equation}
This is clear from assuming the most advantageous distribution of non-zeros/disjoint $d$-closures. This distribution has a non-zero in the first and the $N$th element 
whereas the other non-zeros are equally spaced with distance $2d+1$.

\subsection{d-coherence}\label{sec:sub_d_coherence}
We base a first analysis of \hbox{$d$-tol}\-erant recovery on the notion of coherence. This measure is computationally tractable and a proxy for other measures such as 
the 
restricted isometry/orthogonality property \cite{Ben-Haim2010}. Furthermore, as opposed to the latter, matrices with a specific coherence structure 
can be easily crafted.

The linear sensing model, \eqref{eq:linear_sensing_model_w_noise}, connects the allowed discrepancy in the indices of the recovered non-zeros to the correlation of 
matrix columns with respect to their index distance.

The correlation of any two columns $\phi_i, \phi_j$ of a matrix $\Phi$ can be expressed as:
\begin{equation}
  \mu(i, j) := \mu(\phi_i, \phi_j) = \frac{\abs{\scalarP{\phi_i}{\phi_j}_2}}{\norm{\phi_i}_2\norm{\phi_j}_2}\,. \label{eq:def_shorthand_mu}
\end{equation}
The overall maximum correlation of matrix columns is captured by the coherence of a matrix.

\begin{df}
  \label{def:mu}
  The coherence of a matrix $\Phi$ is defined as 
    \begin{align}
      \mu(\Phi) := \max_{i\neq j} \mu(\phi_i, \phi_j)\;. \label{eq:def_coherence}
    \end{align}
\end{df}

The Welch bound, $\mu(\Phi)\geq\muW(\Phi):=\sqrt{\frac{N-M}{M(N-1)}}$, is the lowest possible coherence for a \lN{2}-column normalized matrix $\Phi\in\C^{M\times N}$, 
see Theorem 5.7 in \cite{Foucart2013Book}. For the Fourier matrix $\muW$ is obtained through row selection from a cyclic difference set \cite{Yu2010}.
If $\mu(\Phi)$ is close to the Welch bound, we call $\Phi$ incoherent.

To analyze \hbox{$d$-tol}\-erant recovery, we extend the notion of coherence to be made dependent on the column index distance.

\begin{df}
  \label{def:mu_d}
Define the set of $d$-spread coefficients (with wrapping) as $$\Gamma_d := \{(i, j)\in[N]^2 : |i-j|>d, |i-j-N|>d, |i-j+N|>d \}.$$
\noindent Then the $d$-coherence of a matrix $\Phi$ is defined as 
    \begin{align}
      \mu_d(\Phi) := \max_{(i,j)\in\Gamma_d} \mu(\phi_i, \phi_j)\;. \label{eq:def_d_coherence}
    \end{align}
  If $\mu_d(\Phi)$ is close to the Welch bound\footnote{Of course, by effectively removing columns from the calculation of coherence, we expect the Welch bound to be slightly weaker. Since we typically consider $d$ to be much smaller than the other parameters, we leave it as-is for simplicity.} for a certain $d$, we call $\Phi$ $d$-incoherent.  
\end{df}

For $d=0$ the definition of $\mu_d(\Phi)$ coincides with that of the coherence. As $d$ increases $\mu_d(\Phi)$ decreases monotonically. Indeed, suppose $f < d$. Then
\begin{align*}
\mu_d(\Phi) 
  = \max_{(i,j)\in\Gamma_d} \mu(\phi_i, \phi_j)
  < \max_{(i,j)\in\Gamma_f} \mu(\phi_i, \phi_j)
  = \mu_f(\Phi),
\end{align*}
where the inequality holds because every $d$ separated set is also $f$ separated, i.e. $\abs{\Gamma_d} < \abs{\Gamma_f}$ with cardinality $\abs{\cdot}$.

With \hbox{$d$-co}\-herence we can ensure that large column correlations are confined to column indices in the $d$-closure of the reference column. This leads to two 
key aspects for any successful \hbox{$d$-tol}\-erant recovery:
\begin{enumerate}
  \item A large \hbox{$d'$-co}\-herence for all $d' < d$ increases noise stability by increasing the number of "distorted copies" of any reference column.
    \label{item:desire_ti_d_incoherence}
  \item A minimal $\ti{d}$-coherence for all $\ti{d} > d$ ensures reconstruction of any support elements of mutually disjoint $d$-closures.
    \label{item:desire_d_coherence}
\end{enumerate}

\begin{figure}
  \centering
  \includegraphics[width=0.6\columnwidth]{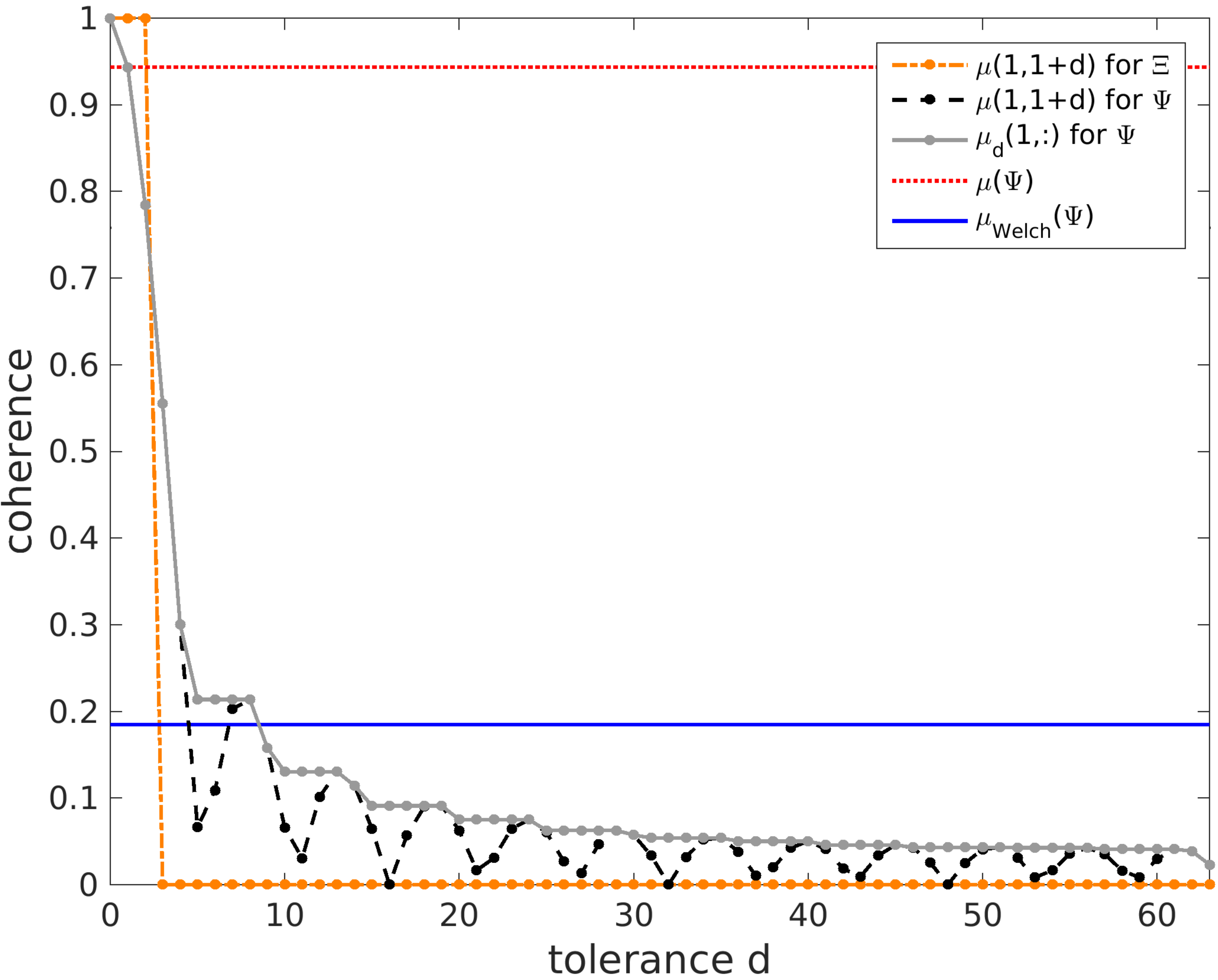}
  \caption{
  The advantageous fast decrease of column correlations, $\mu(i, i+d)$, of the
first and (1+d)th column for $\Psi$ (dashed, black) and $\Xi$ (dash-dotted,
orange) highlights the suitability of those matrices for d-tolerant recovery
which is based on $\mu_d$ (solid, gray; envelope). The matrices are defined in
Section III and were instantiated with $M=24, N=128$. The coherence $\mu(\Psi)$
(dotted, red) and Welch bound $\mu(\Psi)$ (solid blue) are given as references
towards classical CS and the least coherence possible for a matrix of these
dimensions. Symmetric boundary conditions apply for omitted columns.
    \label{fig:correlation_measures}}
\end{figure}

To provide explicit examples, we consider two matrices, $\Xi$ and $\Psi$, defined as follows.  Let $\Psi\in\C^{M\times 
N}$ be the sensing matrix that equals $F_N$ restricted to the first $M$ rows, and let $\Xi\in\C^{M\times N}$ be the sensing matrix that equals $F_M$ inflated by 
$(2d+1)$; in other words, every column of $F_M$ is copied $2d$ times to form a consecutive block of $(2d+1)$ columns in $\Xi$. In \figref{fig:correlation_measures} direct column correlations 
$\mu(1, 1+f)$ of the first with the $(1+f)$th column are shown for $\Psi$ (dashed, black) and $\Xi$ (dash-dotted, orange). Note, only half of the range is 
shown, as the other half is mirrored. Since both $\Psi$ and $\Xi$ are directly derived from the Fourier matrix, they inherit the invariance property of the column 
correlations with shifting reference index $j$,
\begin{equation}
 \forall j\in\bracks{N}:\forall f\in\N_0\;\text{such that}\;\mu(1, 1+f) = \mu(j, j+f)\;. \label{eq:fourier_matrix_invariance}
\end{equation}
Thus the shown correlation pattern is exemplary for any column index. For $\Xi$ this is only true for every $(2d+1)$th column.

If $\Xi$ is constructed with a fixed $(2d+1)=3$ inflation, we observe $3$ columns with $\mu(1, 1+f)=1$. Those are the first column and its copies. For all columns 
further away $\mu(1, 1+f)=0$ due to the orthogonality of the columns in $F_M$. Since it is at least for every $(2d+1)=n$th column $\mu(n, n+d)=\mu_d(\Xi)$ for this 
matrix construction, we have maximal \hbox{$d>d'$-co}\-herence and minimal $d<\ti{d}$-incoherence.  
 To incorporate noise robustness, large 
\hbox{$d'$-co}\-herences that are still unequal to $1$ are preferable. The greater the deviation from $1$, the larger the noise tolerance. This statement is however 
limited. Allowing for too much noise compensation would allow a \hbox{$d$-tol}\-erant reconstruction to completely fail. Experimentally it was found that for OMP a 
$d$-coherence 
larger than $0.75$ is beneficial.

Due to row restrictions from $F_N$ as one continuous block in the case of $\Psi$, we see that large column correlations are possible that are not equal to $1$. Since 
the coherence $\mu(\Psi)$ (dotted, red) is large, from the perspective of conventional CS theory this matrix seems not to be suited for reconstruction. Matrices used in 
CS are usually required to have a coherence that is close to the Welch bound, $\muW(\Psi)$ (solid, blue). We can see however that although correlations of neighboring 
columns in $\Psi$ are large, the level of correlation rapidly drops with increasing distance between the regarded columns. This means the matrix is only partially 
coherent and well suited for \hbox{$d$-tol}\-erant recovery. For $d>8$ (for $M=24,N=128$) it is even $\mu_d(\Psi)<\muW(\Psi)$ motivating the hope that if existing 
incoherent theory could be adapted to $d$-incoherent theory in a similar way, then it would be possible to get an even better performance in \hbox{$d$-tol}\-erant 
recovery than 
the 
incoherent theory would allow for a perfectly incoherent sensing matrix. This behavior is well captured by the \hbox{$d$-co}\-herence (solid, gray). So more 
specifically $\Psi$ is \hbox{$\ti{d}$-in}\-coherent with $\ti{d}>8$ and could be considered \hbox{$d'$-co}\-herent for $d'\leq3$.

Qualitatively this means that in theory, noise robust reconstruction of $S$-sparse signals with small dynamic range, up to an SNR of $0.87$ (equal to $0.45$ of linear 
independence) with $d=3$ and $S=S_{\text{max}}=18$ from $M=24$ measurements would be possible.
In numerical experiments based on OMP and complex valued signals with arbitrary range, this translates into a $3$-tolerant recovery of $6$ more non-zeros on average by 
using the coherent matrix $\Psi$ instead of an incoherent matrix (random row restricted submatrix of $F_N$ of size $M\times N$). To recover at least the same 
amount of non-zeros with incoherent matrices as with partially coherent matrices and $d=3$, the tolerance would have to be increased to $d\geq 8$. This is true 
for any SNR in the range of $[0,\infty]$.

\subsection{Additional definitions}\label{sec:sub_additional_definitions}
In this section we introduce a collection of other important concepts that help characterizing the \hbox{$d$-tol}\-erant recovery setup. We begin with generalizing the 
concept 
of the aforementioned column correlation invariance, \eqref{eq:fourier_matrix_invariance}, of Fourier submatrices obtained by row selection. The distribution of highly 
correlated columns within $\Phi$ can be characterized in terms of matrix coherence functions.
\begin{df}
  \label{def:matrix_coherence_functions}
  The set of matrix coherence functions $\setbr{\mu^{(j)}}_{j\in\bracks{N}}$ of a matrix $\Phi\in\C^{M\times N}$ is defined through
  \begin{align}
    \mu^{(j)} := \paren{\mu\paren{\phi_j,\phi_1},\dots , \mu\paren{\phi_j,\phi_N}} \,. \label{eq:def_matrix_coherence_functions}
  \end{align}
\end{df}

With the help of the matrix coherence functions, two fundamentally different types of partially coherent matrices can be distinguished.

\begin{df}
  \label{def:classes_of_coherence_neighborhoods}
  A set of matrix coherence functions is called dynamic, if the correlation of any column with the reference column depends only on the difference of the column 
  indices. Otherwise a set of matrix coherence functions is called static.
\end{df}
The choice of the terminology static and dynamic is motivated by the simple cases (i) when all neighboring columns are highly correlated the coherence functions can 
be viewed via the gram matrix $\Phi^*\Phi$ and appear as a sliding gradient (dynamic), e.g. $\Phi\equiv\Psi$, whereas (ii) when the matrix contains blocks of correlated 
columns and columns in different blocks are uncorrelated, the gram matrix consists of a rigid series of blocks (static), e.g. $\Phi\equiv\Xi$.

A similar \hbox{$d$-tol}\-erant extension as was made to the coherence can be made to the cumulative coherence (also known as \lN{1}-coherence or the Babel function). 
It will be used in the proof of Theorem~\ref{thm:dtompv2_dtrc_guarantee}. The cumulative \hbox{$d$-co}\-herence is one way to quantify the correlations of any given 
element with a consecutive, disjoint block of length at most $2d+1$.

\begin{df}
  \label{def:mu_c_d}
  For $\Phi\in\C^{M\times N}$, we define its cumulative \hbox{$d$-co}\-herence $\mu^C_d(\Phi,k)$ with test-set cardinality $k$ as:
    \begin{align}
      \mu^C_d(\Phi, k) &:= \max_{\substack{\Gamma\subset\bracks{N}\\\abs{\Gamma}=k}}\max_{i \notin \clos{d}{\Gamma}} 
      \sum_{j\in\Gamma} \frac{\abs{\scalarP{\phi_i}{\phi_j}_2}}{\norm{\phi_i}_2\norm{\phi_j}_2} \,. \label{eq:def_cummulative_d_coherence}
    \end{align}
  We write $\mu^C(\Phi, k) := \mu^C_0(\Phi, k)$ for the standard cumulative coherence.
\end{df}

It is easy to see that the cumulative $d$-coherence satisfies the following properties:
\begin{itemize}
  \item $\mu^C_d(\Phi,k)$ is monotonically decreasing as $d$ increases. Indeed, we have for any $f<d\in\N_0$:
  \begin{equation}
    \mu^C_d(\Phi,k) \leq \mu^C_f(\Phi,k) \leq \mu^C(\Phi,k)\,.     \label{eq:mu_c_d_monotony_wrt_d}
  \end{equation}
  \item $\mu^C_d(\Phi,k)$ is monotonically increasing as $k$ increases:
    \begin{equation}
      \forall k<l\in\N_0: \mu^C_d(\Phi,k) \leq \mu^C_d(\Phi,l)\,.     \label{eq:mu_c_d_monotony_wrt_k}
    \end{equation}
  \item The lower bound given in Theorem 5.8 of \cite{Foucart2013Book} applies by replacing $N$ by $\hat{N}:=\max\setbr{M, \ceil{\frac{N}{d}}}$. That is:
    \begin{equation}
      k \leq \sqrt{\hat{N}-1} \lra \mu^C_d(\Phi, k)\geq k\sqrt{\frac{\hat{N}-M}{M(\hat{N}-1)}}\,.
    \end{equation}
\end{itemize}

\section{Numerical simulation results}\label{sec:numerical_simulation_results}
In this section we demonstrate the advantage of coherence in \hbox{$d$-tol}\-erant recovery using numerical simulation results. The main part of the results is 
based on the \hbox{$d$-tol}\-erant recovery measure, \eqref{eq:def_dtol_recovery}. Results shown in \figref{fig:comparebasic} and \figref{fig:compare_convex} were 
 produced using MATLAB 2017a \cite{matlab} with the median of 100 iterations per data point. All other results shown below are obtained via standard OMP with MATLAB 
2015b \cite{matlab} 
using the RPECS Matlab toolbox (version 1.1) \cite{RPECS} and with the median of $500$ iterations per data point. For each iteration just the signal and the 
noise were re-initialized. The matrices were newly initialized for each set of parameters only. All generated signals were complex valued. Their support was uniformly 
random distributed. The amplitudes of the real and imaginary parts were selected i.i.d. for every non-zero, uniformly at random on $[-50,50]$. The noise entries were 
i.i.d. standard normally distributed and then rescaled to fit the desired signal-to-noise ratio (SNR). Note that we do \textit{not} force the signal to have spread 
support unless explicitly stated.

We consider several types of sensing matrices, given as:

\begin{description}
\item \Fcb: The first $M$ consecutive rows of $F_N$. (Called $\Psi$ in Section~\ref{sec:sub_d_coherence}.)
\item \Fc: Any $M$ consecutive rows of $F_N$. The shift of the block was uniformly random distributed and selected from $\bracks{1,N-M}$.
\item \Fr: $M$ rows of $F_N$ were selected uniformly at random.
\item \Fx: Uses $5\floor{\log{N}}$ blocks of consecutive columns from \Fc, constructed from $F_{3N}$.
\item \Rg: Gaussian random matrix.
\end{description}
All matrices were \lN{2}-column normalized.

The incoherent matrices are \Fr~and \Rg, and the partially coherent matrices are \Fcb, \Fc~and \Fx. 
\Fx~is an approximation to a matrix with static matrix coherence function. The coherence across the column blocks will be low and thus the matrix will 
appear to have almost rigid blocks of high coherence.

\begin{figure*}
\centering
  \subfloat[$N=1024$\label{fig:contour_plot_1024}]{
  \includegraphics[width=0.49\columnwidth]{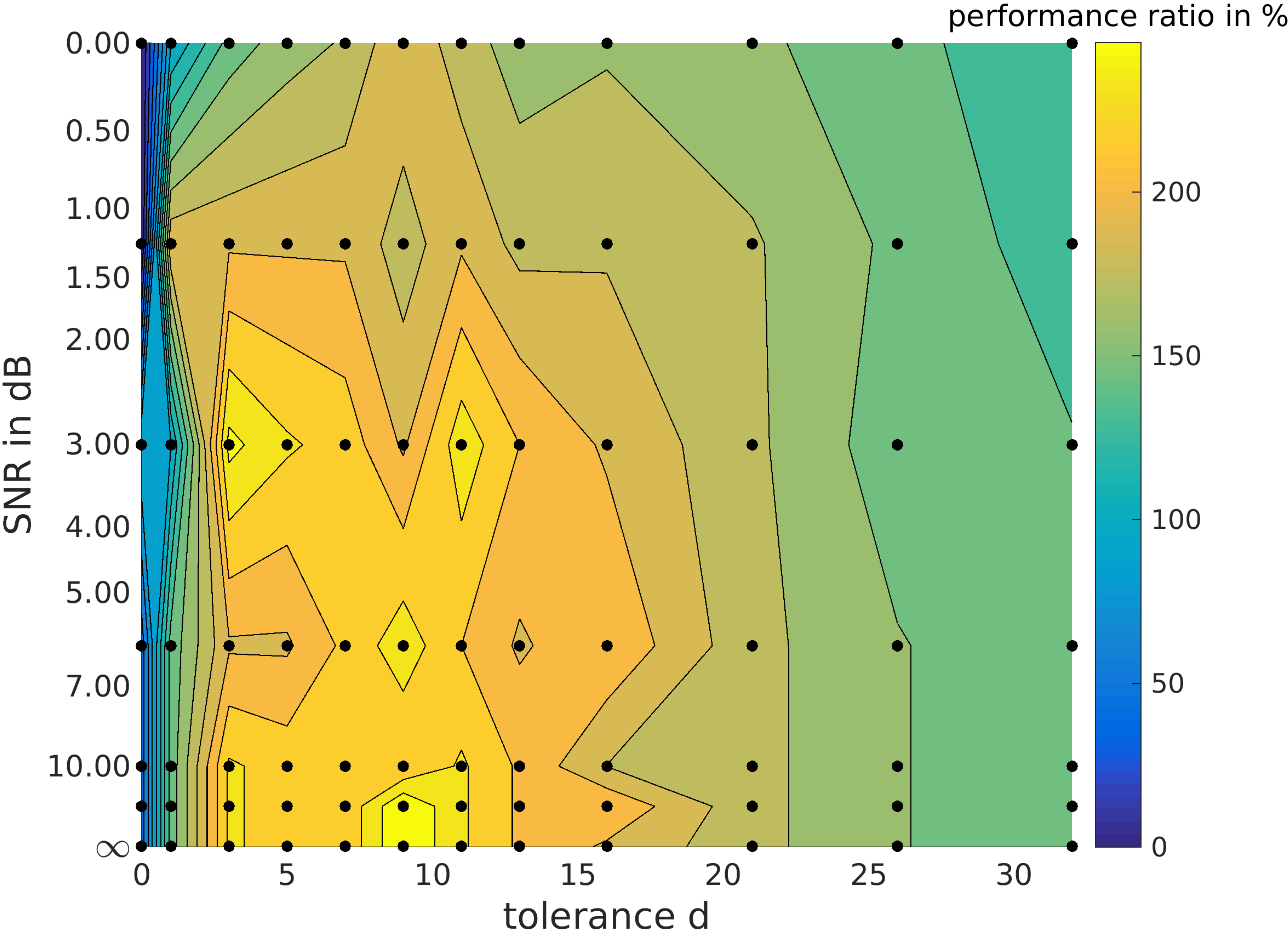}
  }
  \subfloat[$N=2048$\label{fig:contour_plot_2048}]{
  \includegraphics[width=0.49\columnwidth]{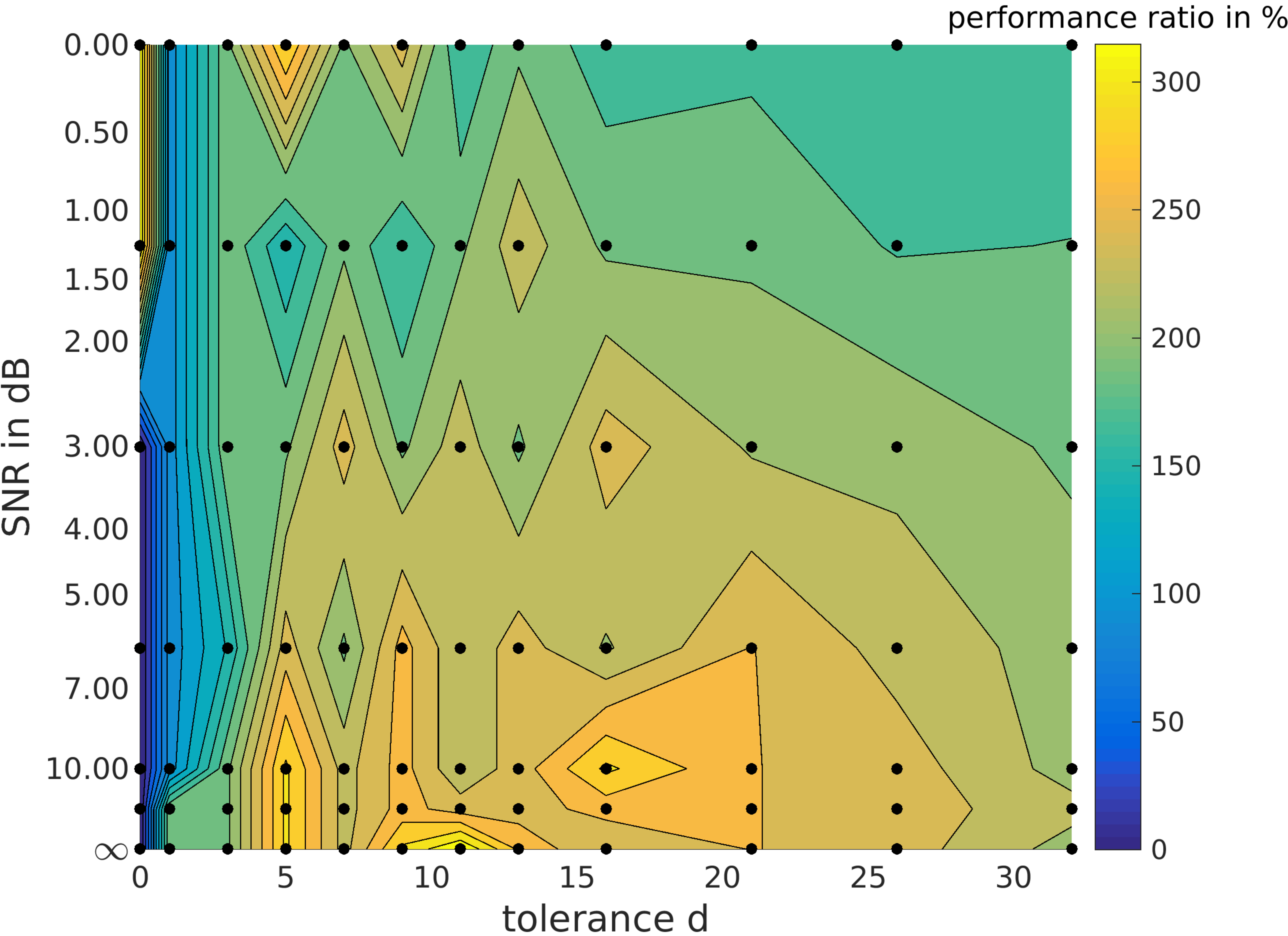}
  }
  \caption{Consider the percentage of non-zeros recovered with respect to the \hbox{$d$-tol}\-erant recovery measure. Let $X$ be the percentage recovered with a 
  partially coherent sensing matrix (\Fc) and $Y$ the percentage recovered with an incoherent sensing matrix (\Fr). The plot shows the ratio $X/Y$ and as such the benefit of 
  coherent sensing for varying tolerances $d$ and amounts of noise. The number of measurements, $M=32$, the sparsity, \hbox{$S=16$}, and the signal dimension, $N$, 
  are fixed within each plot. 
  \label{fig:contour_plot}}
\end{figure*}

Shown in \figref{fig:contour_plot} are the ratios of the percentages of non-zeros that could be recovered, \hbox{$d$-tol}\-erant wise, with a partially coherent 
sensing matrix (\Fc) over an incoherent sensing matrix (\Fr). The color bar represents the ratio in recovery percentages; thus, when the color is greater than 
$100\%$ we see improvements with our method. The presented situations are heavily undersampled with the number of measurements $M=32$ fixed and $N=1024$ or $N=2048$. 
In both plots, we see an optimal value for the tolerance $d$. More importantly, we observe improvements from coherence (i.e. when the ratio percentage is above 
100) for a broad range of values of $d$, especially in the mild SNR regime. Only if $d=0$ incoherent matrices perform like partially coherent sensing matrices. 
Unsurprisingly, this means that for exact support recovery the incoherent 
sensing is similar or better (for small undersampling factors $N/M\leq4$, see discussion of \figref{fig:m_trend} below) in determining the position of every 
non-zero.

\begin{figure*}[ht!]
  \centering
  \subfloat[$N/M=1024/32$\label{fig:d_trend_1024}]{
  \includegraphics[width=0.49\columnwidth]{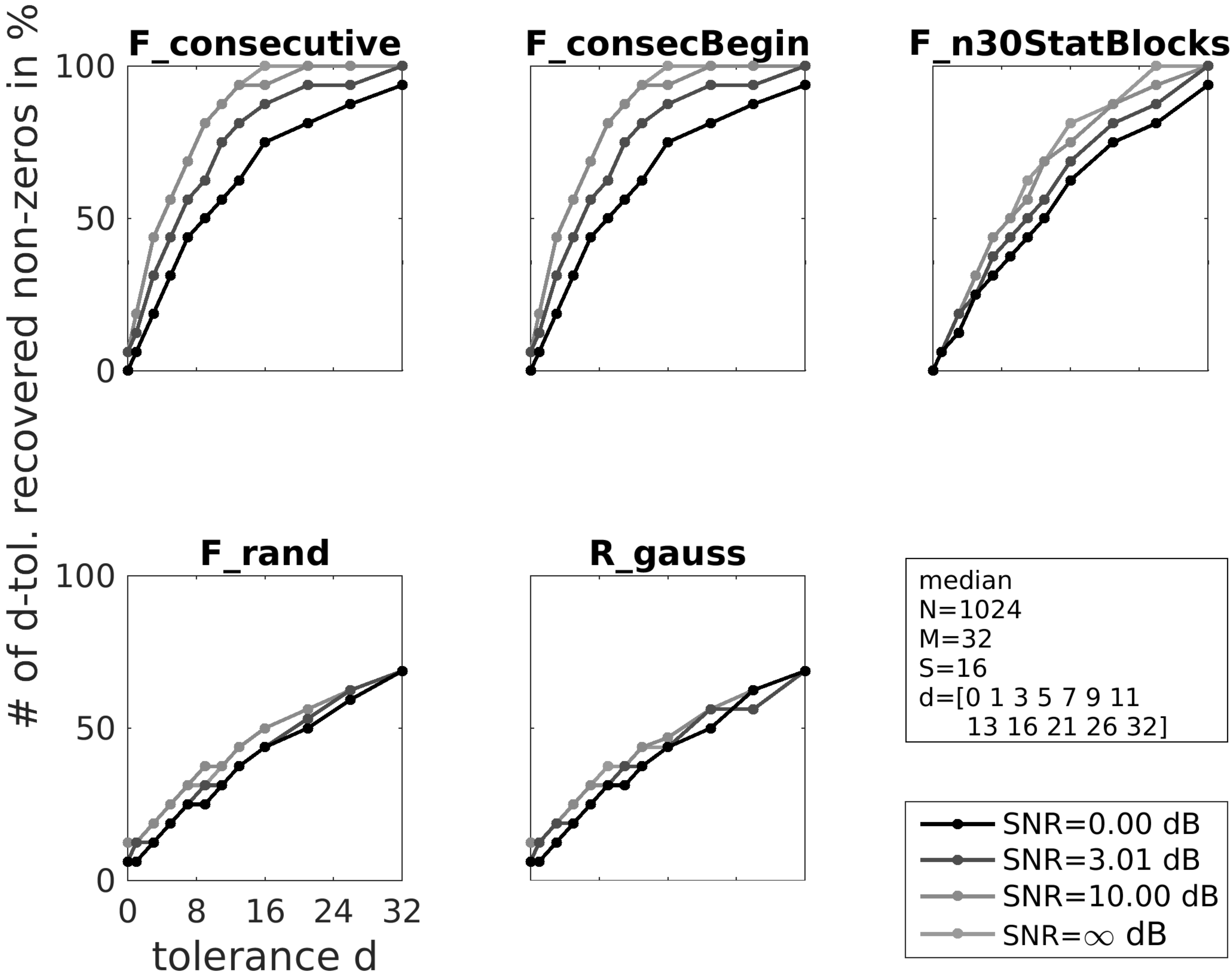}
  }
  \subfloat[$N/M=2048/64$\label{fig:d_trend_2048}]{
  \includegraphics[width=0.49\columnwidth]{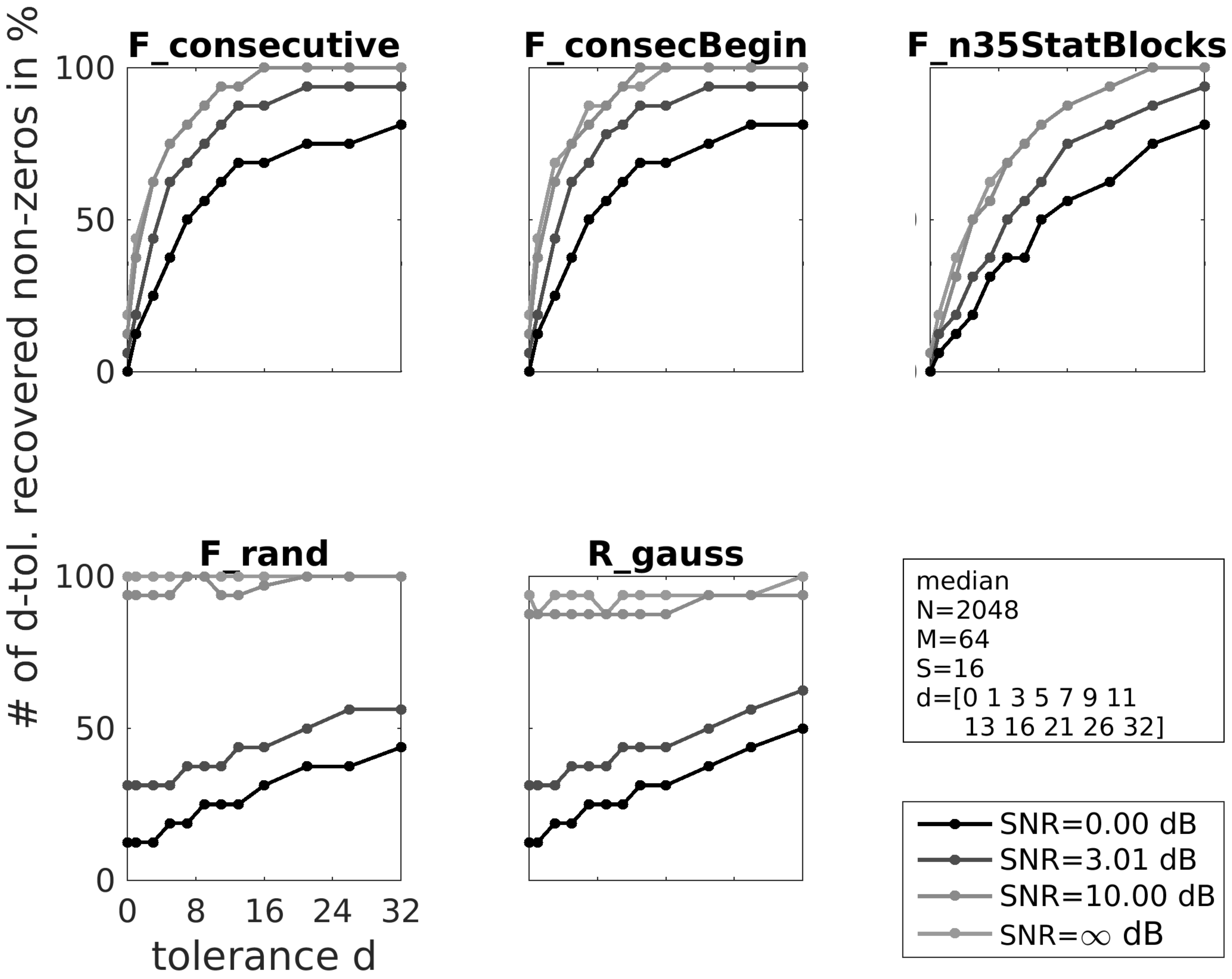}
  }
  \caption{
The advantage of coherent (top) over incoherent sensing matrices (bottom) is
illustrated in terms of percentages of \hbox{$d$-tol}\-erant recovered non-zeros via OMP
for an increasing tolerance $d$ and several noise levels. The undersampling
factor is fixed at $N/M = 32$ and the sparsity is $S = 16$.
  \label{fig:d_trends}  
  }
\end{figure*}

\hbox{Figure~\ref{fig:d_trends}} shows the average number of \hbox{$d$-tol}\-erantly recovered non-zeros as a function of $d$, for various types of sensing matrices, 
various noise 
levels, and only few measurements ($N/M=32$). We consider $N/M=1024/32$ in \figref{fig:d_trend_1024} first. The partially coherent sensing matrices \Fc~and 
\Fcb~(with dynamic matrix coherence functions) perform especially well. Most importantly, already with small values of $d$ ($\approx 5$) much more non-zeros can be 
reconstructed. If little noise is present ($\snr\leq10$dB), for $d=9$ the number of reconstructed non-zeros is doubled for partially coherent versus incoherent sensing. 
Close to $100\%$ recovery is reached for $d>16$ in the low noise setting. 
\hbox{Figure~\ref{fig:d_trend_2048}} with $N/M=2048/64$ shows what happens if both the signal dimension and the number of measurements get scaled up. Due to the lower 
normalized sparsity $S/M=0.25$ incoherent sensing matrices are able to perform well for large SNR's ($\geq10$dB). As the amount of measurement noise 
increases the incoherent matrices are however drastically impacted ($31\%$ instead of $\approx100\%$ for $d=0$, $\snr=3.01$dB for both incoherent matrices). The impact 
of noise on the coherent sensing matrices is much less severe especially for $d\geq7$. 
This emphasizes that partially coherent sensing matrices can be employed very effectively at their optimal level of incoherence for challenging signal detection 
situations. The percentage of \hbox{$d$-tol}\-erantly recovered non-zeros is in general monotonic with increasing $d$ amongst all the sensing matrices. Therefore, 
selecting a 
large value for $d$ will not result in substantially worse recovery. This general rule coincides with intuition. For (almost) exact support reconstruction 
($d\in\setbr{0,1}$) using coherence is irrelevant or even bad in the heavily under-sampled setting throughout all SNR levels.

\begin{figure*}[ht!]
  \centering
  \includegraphics[width=0.49\columnwidth]{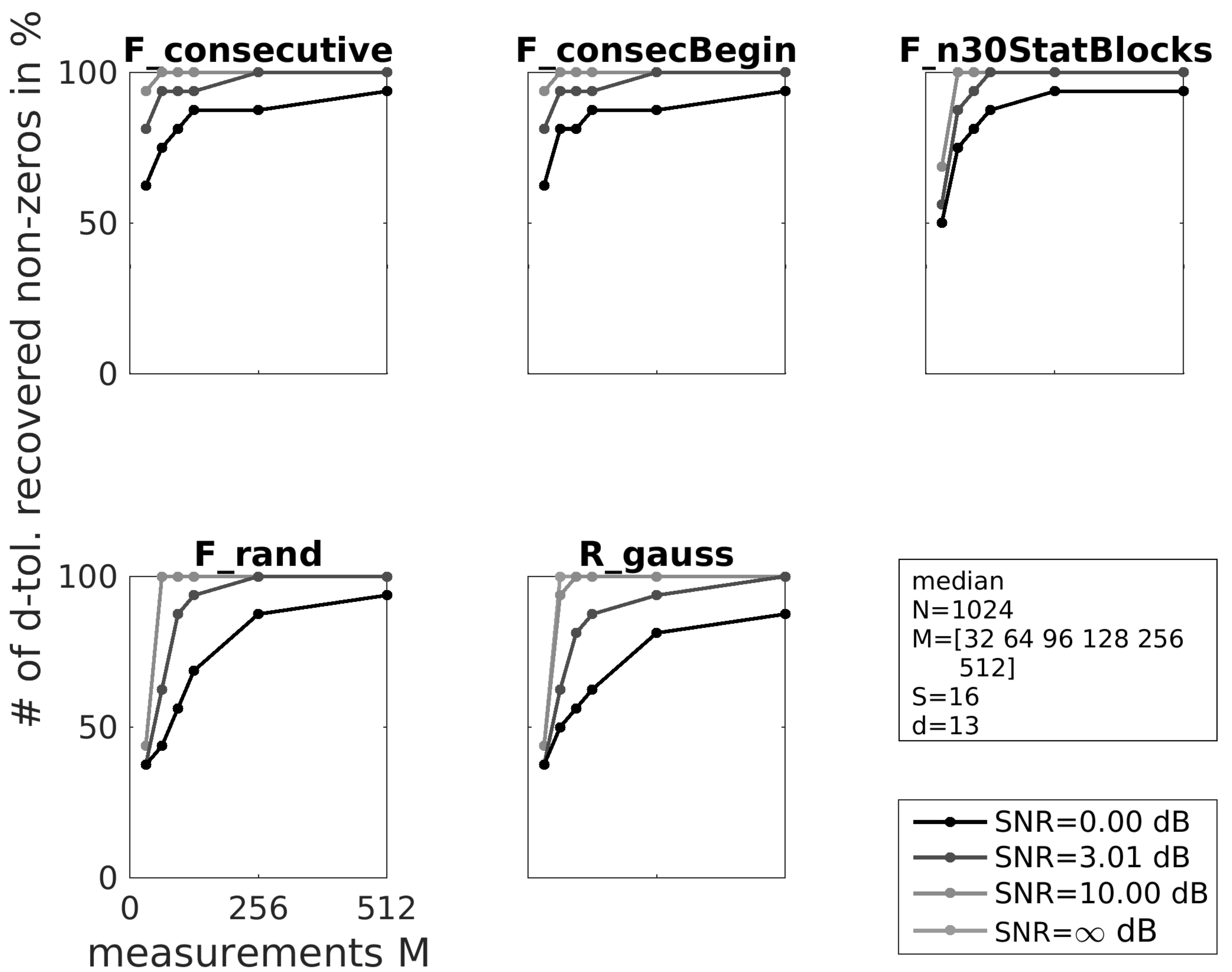}
  \caption{Even with an increasing number of measurements, the coherent (top) dominate the
incoherent sensing matrices with respect to the percentage of d-tolerant
recovered non-zeros via OMP for various noise levels. The signal
dimension, the sparsity, and the tolerance are fixed at $N = 1024$, $S = 16$, and $d = 13$.
  \label{fig:m_trend} }
\end{figure*}

\hbox{Figure~\ref{fig:m_trend}} depicts the trends in \hbox{$d$-tol}\-erant support recovery for an increasing number of measurements $M$ while $N$ is fixed, for 
various types of 
sensing matrices, and for various noise levels. Again \Fc~and \Fcb~perform especially well. Coherent sensing matrices make much better use of additional, possibly very 
distorted, measurements, as soon as a certain tolerance (e.g. $d=13$) in the signal support is allowed. For larger numbers of measurements $M\geq256$ incoherent sensing 
matrices perform similarly well, independent of the SNR. This again underlines that partially coherent sensing matrices are especially interesting for applications in 
which using few measurements is key.

\begin{figure*}[ht!]
\centering
  \includegraphics[width=0.49\columnwidth]{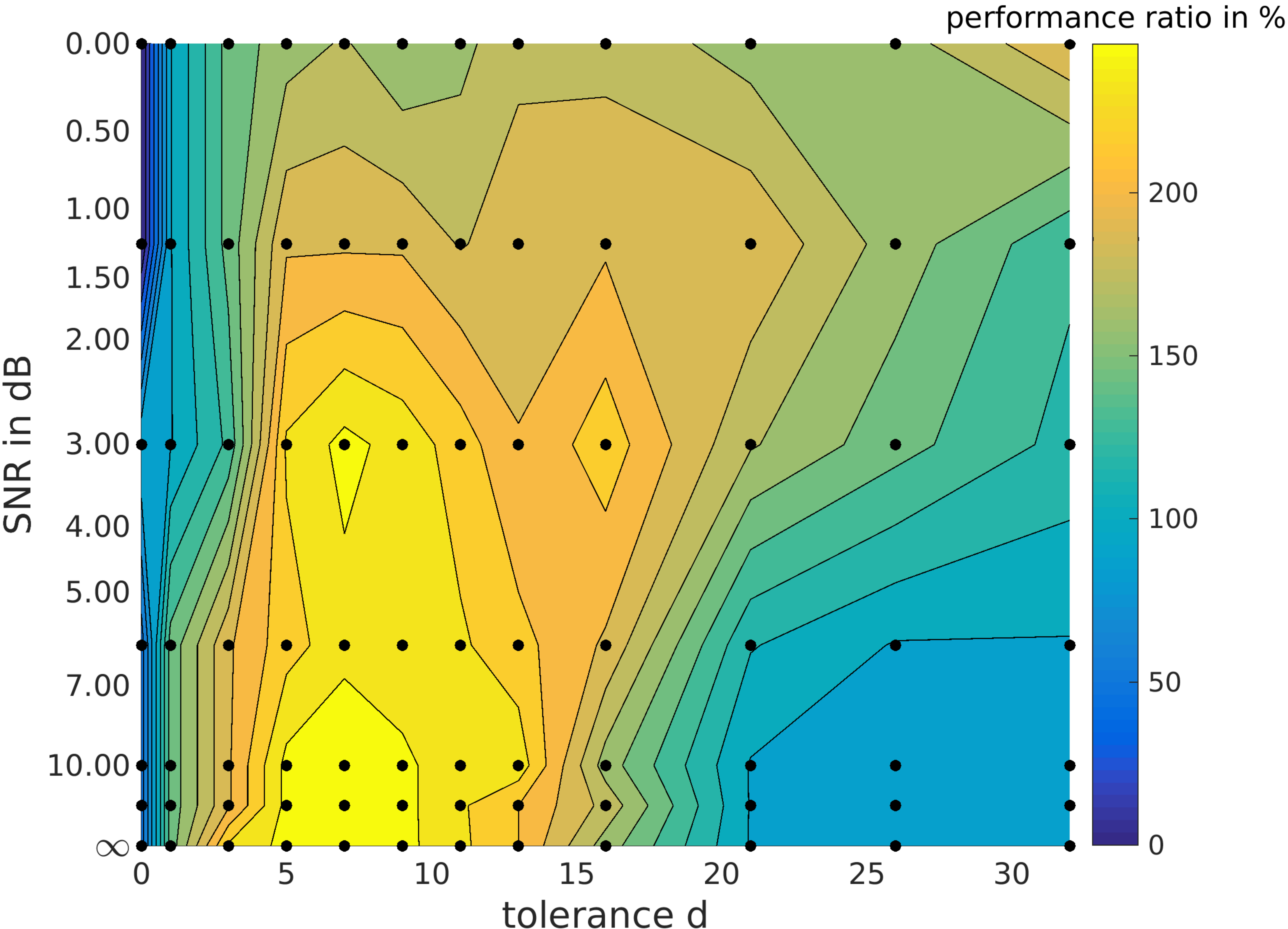}
  \caption{Extended regions with more then twice as many d-tolerant recovered non-zeros
using coherent versus incoherent sensing matrices are shown in the analogue of
\figref{fig:contour_plot_1024} for the special case of signals that have their non-zeros never closer
then $(4d + 1)$.
  \label{fig:contour_plot_1024_spread}}
\end{figure*}
Since we will consider in Section~\ref{sec:theorems} only signals that have their non-zeros never closer then $(4d+1)$, we provide results for those signals in 
\figref{fig:contour_plot_1024_spread}, in analogy to \figref{fig:contour_plot}. 

Next, we compare the coherent sensing paradigm to the simple subsampling strategies described in Section \ref{sec:sub_d_tolerant_recovery}.  Unsurprisingly, the 
second two naive approaches described there yield very poor results and are not even competitive. \hbox{Figure \ref{fig:comparebasic}} displays the results for the 
``Subsampling on coarse grid'' approach; using the standard OMP reconstruction method. The notation $F_R$ indicates the rows were subsampled at random, 
and $F_{cB}$ indicates they were selected to be the first $M$ consecutive rows. Since $d$ is typically much smaller than $N$, both types of sampling approaches are in 
some sense coherent, so it is not surprising that both are somewhat comparable. Our design, however, maintains the signal on a finer grid, which induces slightly more 
coherence, which is evident in the improved reconstruction.

\begin{figure*}
  \centering
  \subfloat[$N/M=1024/64$\label{fig:naive_1024}]{
  \includegraphics[width=0.49\columnwidth
  ]{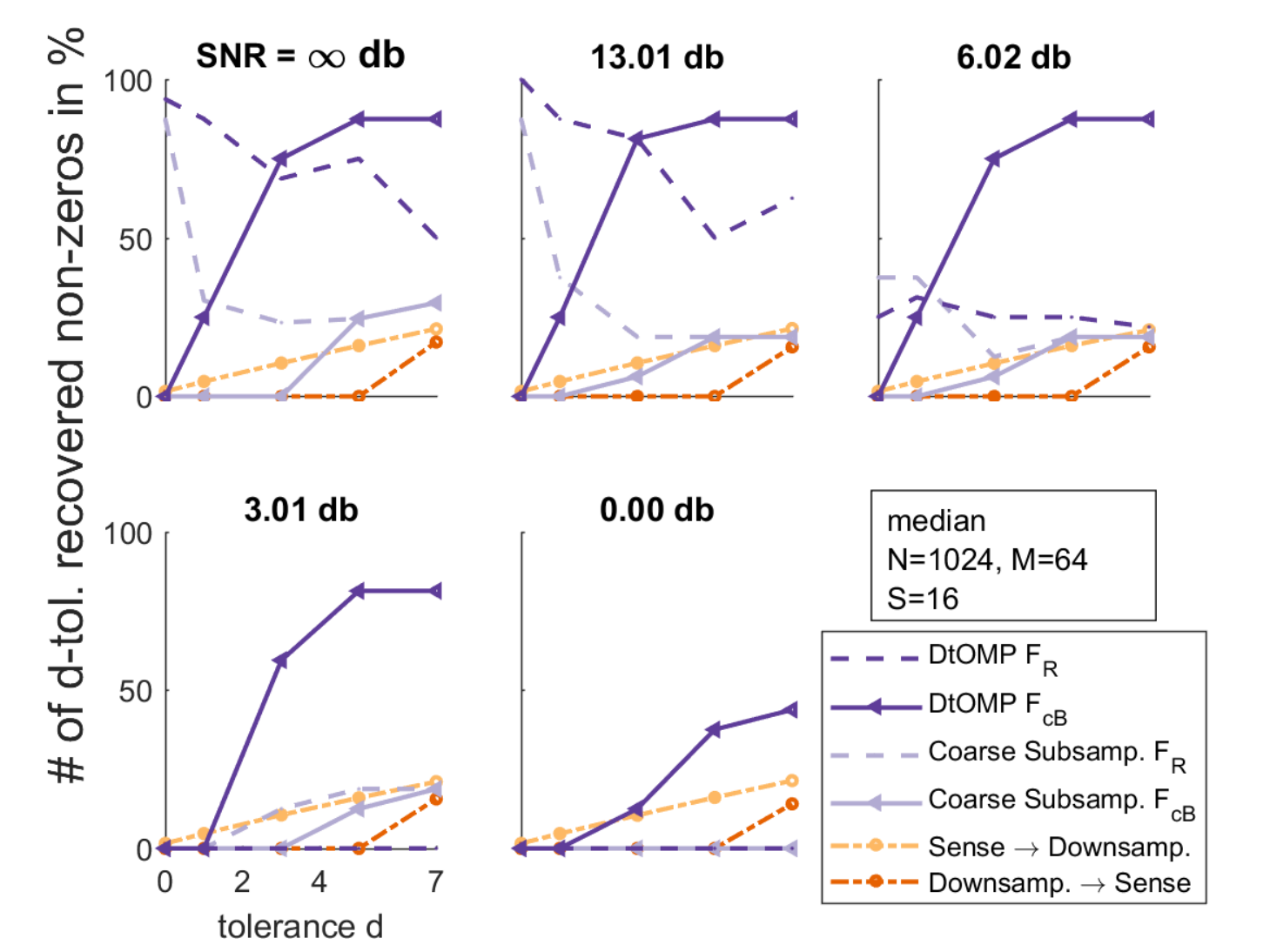}
  }
  \subfloat[$N/M=2048/64$\label{fig:naive_2048}]{
  \includegraphics[width=0.49\columnwidth
  ]{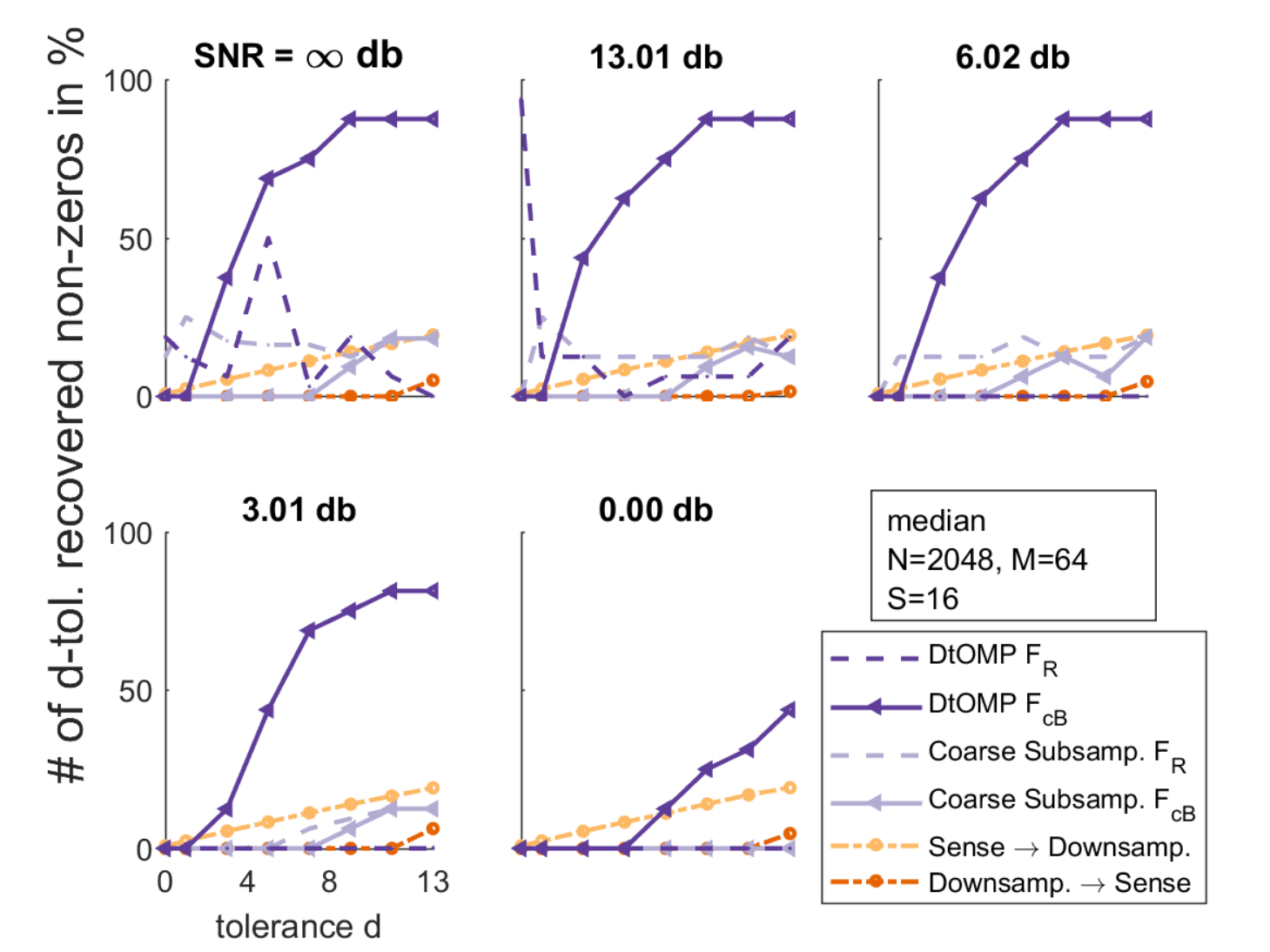}
  }
  \caption{A comparison of the proposed DtOMP, Algorithm~\ref{alg:modified_omp}, with the other methods outlined in Section
\ref{sec:sub_d_tolerant_recovery} reveals that out of these options using the coherent sensing matrix,
$F\_consecBegin$ ($F\_cB$) with DtOMP, recovers the largest percentage of
non-zeros with a certain d-tolerance as soon as considerable noise is present
(top) and/or the undersampling ratio is increased drastically (bottom).
$F\_rand$ ($F\_R$) was the chosen representative for incoherent matrices.
 \label{fig:comparebasic}  
  }
\end{figure*}

In this paper we have focused on greedy methods for simplicity of the analysis, but for completeness we include some results using convex methods for reconstruction, 
\figref{fig:compare_convex}. In particular, we compare the results using the proposed DtOMP, Algorithm~\ref{alg:modified_omp}, against Basis Pursuit Denoising (using 
SPGL1). We see similar trends and behavior in terms of the tolerant \lN{2}-error measure, \figref{fig:compare_convex_intensity}, but from the number of tolerant 
recovered non-zeros its clear that actually often the greedy approach outperforms the convex method, \figref{fig:compare_convex_support}. The reason for this behavior 
are multiple false positives in the case of the convex method. Note especially that the typical partially coherent sensing matrix $\text{F}_{\text{cB}}$ has an 
advantage over the incoherent matrix $\text{F}_{\text{R}}$ when noise is present and tolerant recovery is the objective. 
However, we emphasize once again that the OMP-based reconstruction method is likely still not optimal, and that further study should be done to analyze reconstruction 
performance under this new paradigm of beneficial coherent sensing.

\begin{figure*}
  \centering
  \subfloat[\label{fig:compare_convex_intensity}]{
  \includegraphics[width=0.49\columnwidth
  ]{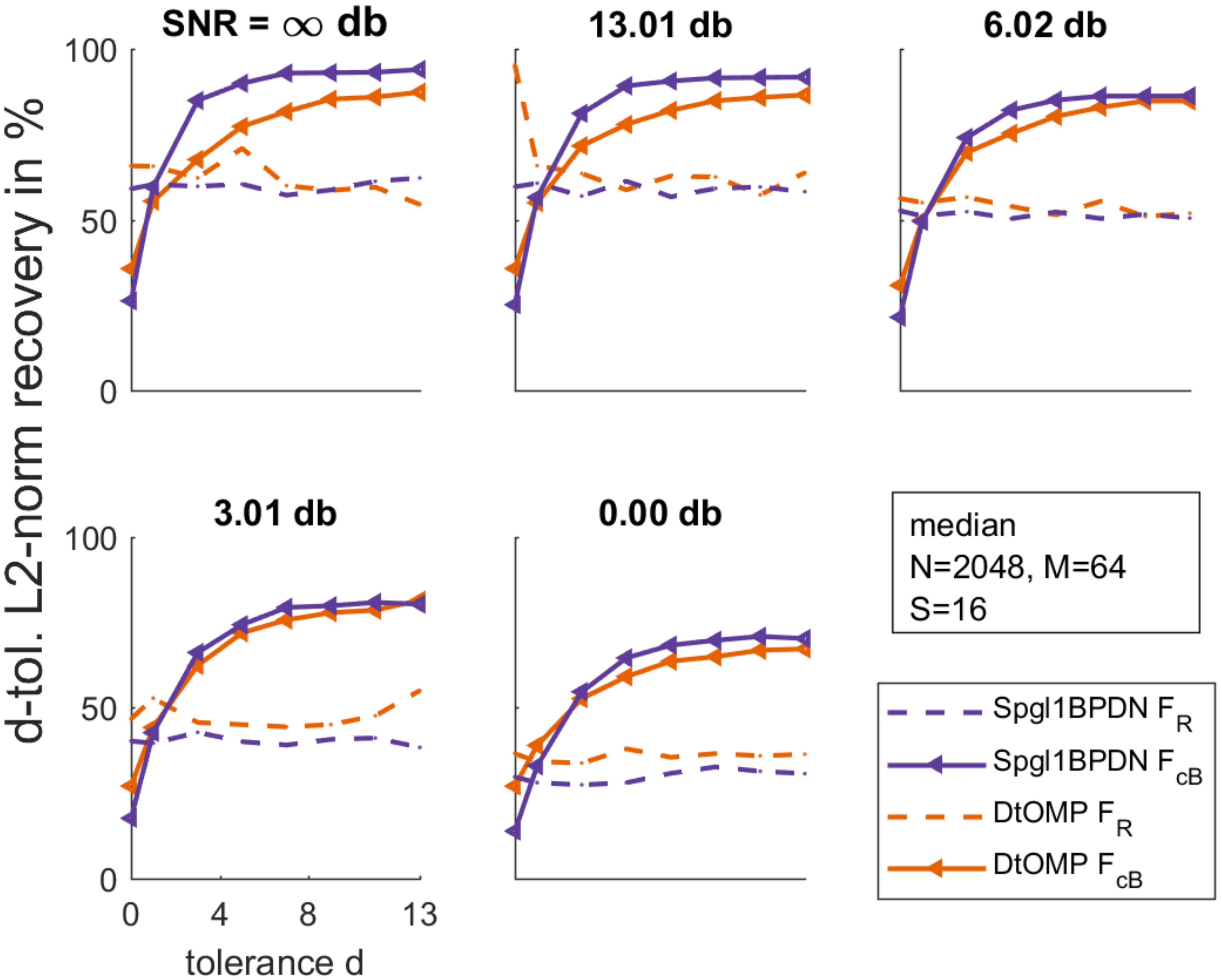}
  }
  \subfloat[\label{fig:compare_convex_support}]{
  \includegraphics[width=0.49\columnwidth
  ]{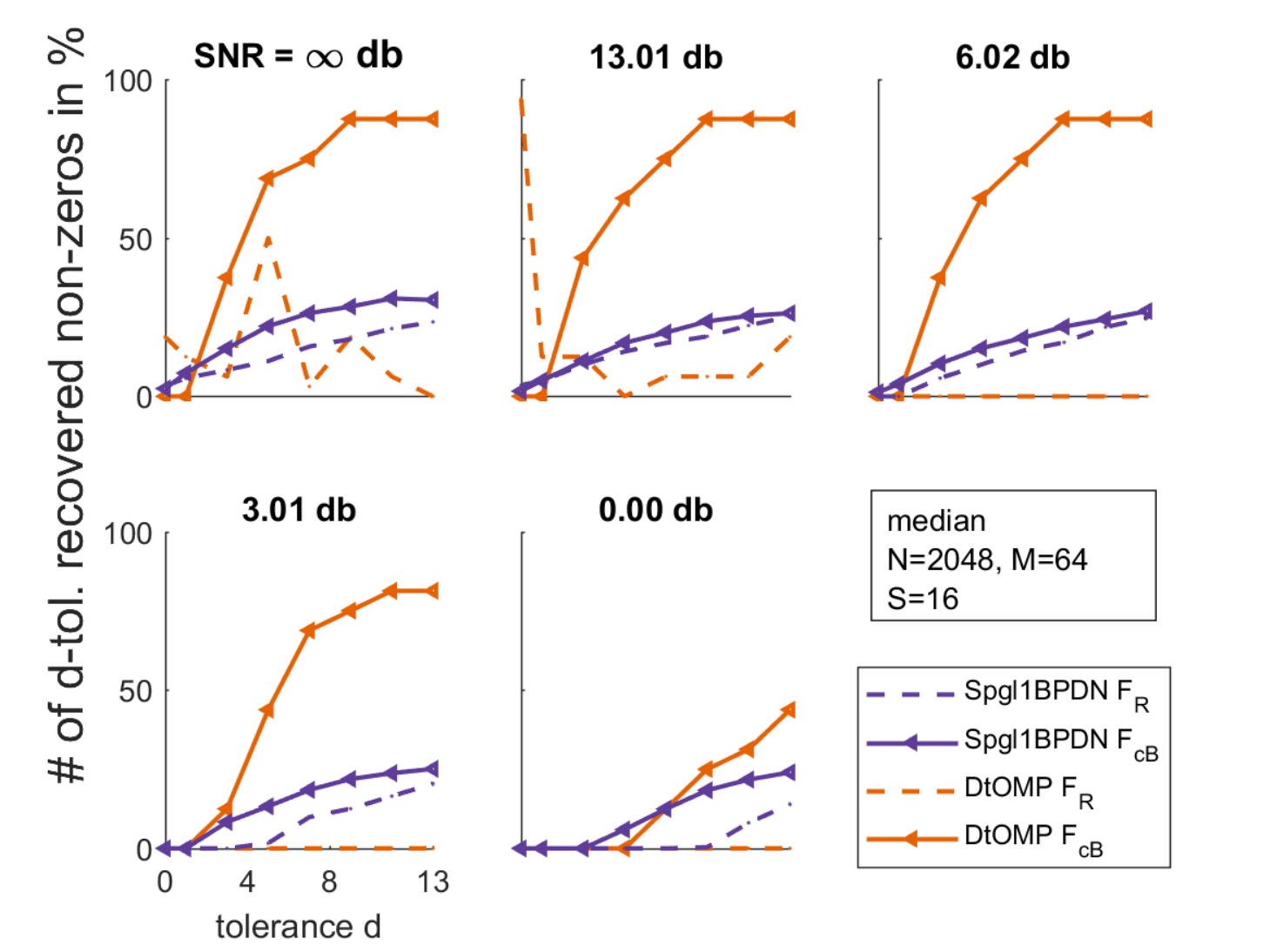}
  }
  \caption{Convex methods such as SPGL1 BPDN are comparable to DtOMP in terms of tolerantly recovered intensity but not in 
  percentages of \hbox{$d$-tol}\-erant recovered non-zeros for large undersampling ratios
  of $N/M = 32$, irrespective of the measurement noise and tolerance chosen.\label{fig:compare_convex} }
\end{figure*}

We close with a remark on the \hbox{$d$-tol}\-erant \lN{2}-norm error based recovery measure $\rho_2$, introduced below. Finding such a measure is not trivial but may 
be desired for classification of the magnitude differences of reconstruction and true signal. We choose a measure that requires knowledge about the true signal and is 
evaluated in two steps: First, we create new proxy signals $x_p, \ti{x}_p$ via: 
\begin{subequations}
  \begin{align}
     (x_p)_i &= \begin{cases}
          \sum_{j\in\clos{d}{i}} \abs{x_j}, &\text{if } i\in\supp{x}\\
          0, & \text{otherwise}
        \end{cases}\label{eq:d_tolerant_l2_recovery_metric_proxy_signal}\\
     (\ti{x}_p)_i &= \begin{cases}
          \sum_{j\in\clos{d}{i}} \abs{\ti{x}_j}, &\text{if } i\in\supp{x}\\
          0, & \text{otherwise}\;.
        \end{cases}\label{eq:d_tolerant_l2_recovery_metric_proxy_recovery}
  \end{align}
\end{subequations}
For example, if $x = (1,0,1,0,0,1,1,0,0)$ and $d=2$, we have $x_p = (2,0,2,0,0,2,2,0,0)$. Note that we sum over the same set of indices in both cases, which causes both proxy signals to share the same support. In a second step we compute
\begin{equation}
  \rho_2(\ti{x}_p, x_p) := 1-\frac{\norm{\ti{x}_p - x_p}_2}{\norm{\ti{x}_p}_2 \norm{x_p}_2}\;. \label{eq:d_tolerant_l2_recovery_metric}
\end{equation}

\begin{figure*}
\centering
  \subfloat[$M=32$\label{fig:d_trend_32_dl2}]{
  \includegraphics[width=0.49\columnwidth]{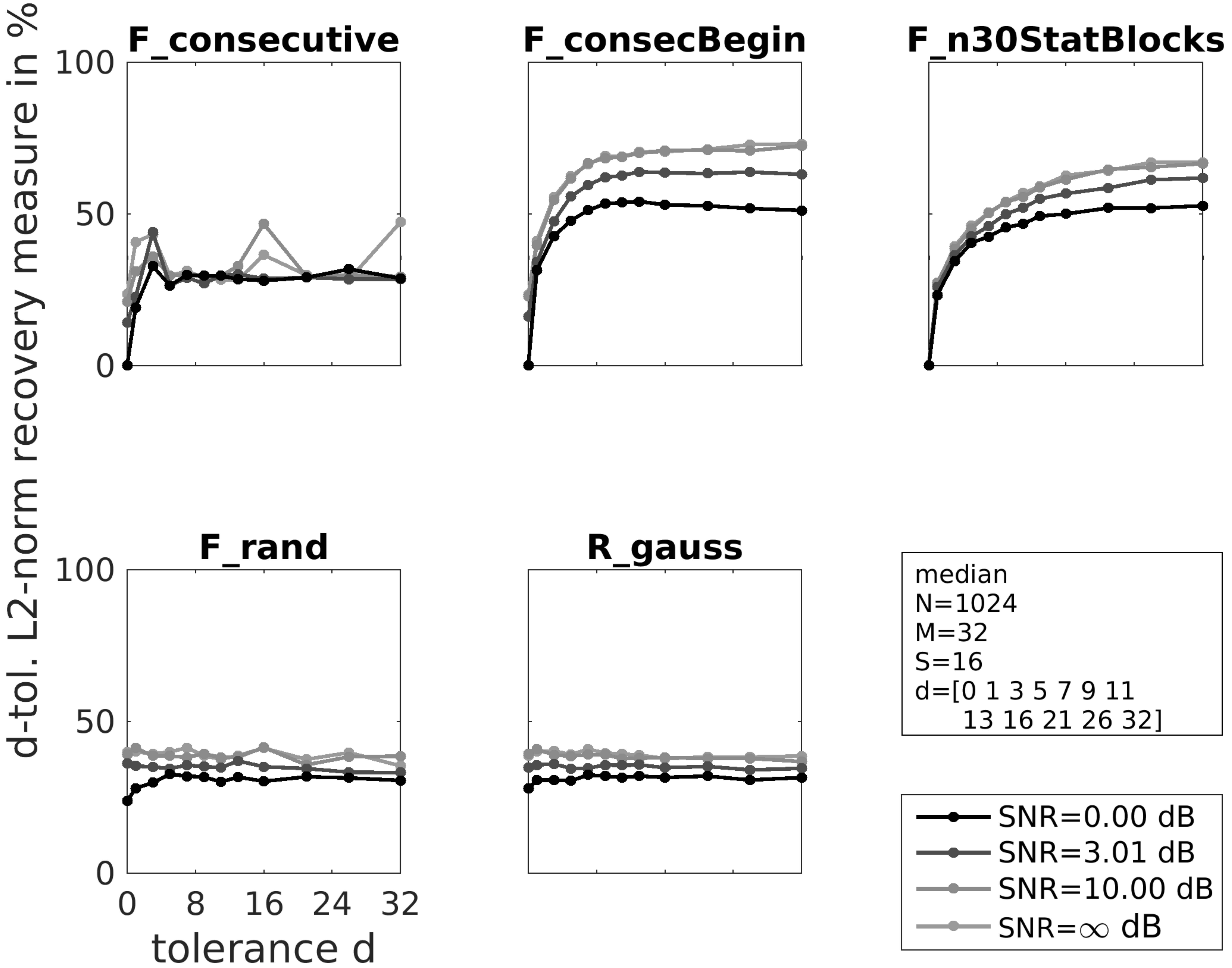}
  }
  \subfloat[$M=64$\label{fig:d_trend_64_dl2}]{
  \includegraphics[width=0.49\columnwidth]{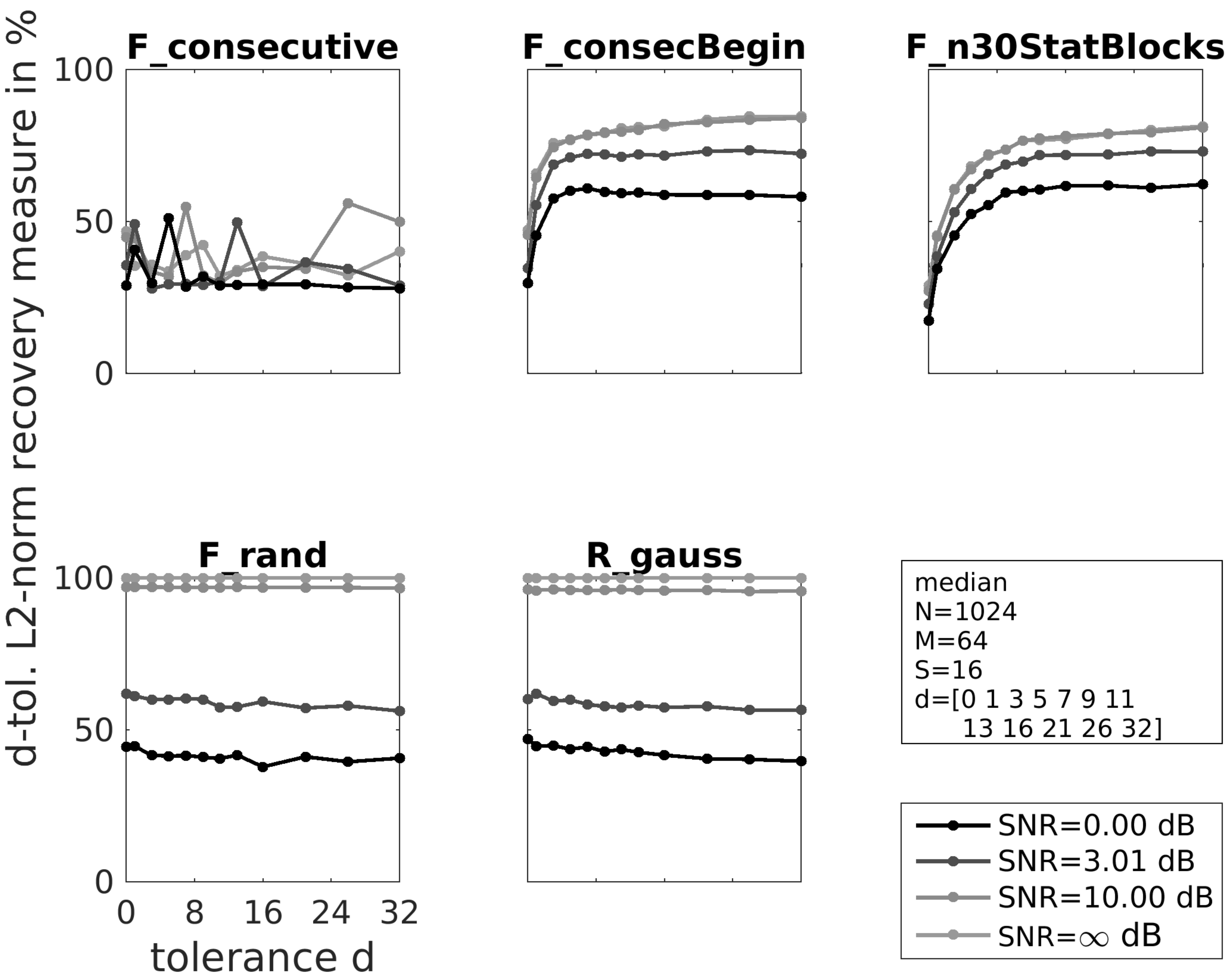}
  }
  \caption{Incoherent (bottom) can compete with coherent sensing matrices (top) in terms
of the \hbox{$d$-tol}\-erant recovered $\ell_2$-norm only if the undersampling ratio
is comparatively low, $N/M=16$, and the measurement noise is low as well. For
harsher sensing conditions $F\_consecBegin$ usually gives the best performance as
well.
    \label{fig:d_trends_dl2}}
\end{figure*}

For that recovery measure we find incoherent sensing matrices are favorable for any SNR and $M\in[64,256)$, i.e. a normalized sparsity smaller than $S/M\leq0.25$ and an 
undersampling factor larger than $N/M\geq4$. This is depicted exemplary in \figref{fig:d_trends_dl2} for $S/M=0.5$ in \figref{fig:d_trend_32_dl2} and $S/M=0.25$ in 
\figref{fig:d_trend_64_dl2}. For \Fcb~and \Fx~we observe the same low impact of measurement noise and about $50\%$ recovery as soon as the tolerance is set sufficiently 
large, $d\geq7$. We notice \Fcb~is slightly better than \Fx~but \Fc~produces a much weaker and highly inconclusive result. This is due to the different construction of 
\Fcb~and \Fc. Both share exactly the same coherence pattern (absolute value) but in general only the former has a smoothly varying phase difference among the real and 
imaginary parts of the columns. The latter experiences rapid phase shifts in real and imaginary part from column to column. Thus an approximate \hbox{$d$-tol}\-erant 
reconstruction can have a quite different magnitude even though reconstructed columns are largely correlated to the true support. Nevertheless, the findings for 
\Fcb~and 
\Fx~stress that it is not only possible to better reconstruct the approximate support but also the approximate magnitude by using coherent sensing matrices in 
difficult sensing scenarios. With a smaller undersampling factor, e.g. $N/M\geq4$, the results (not depicted here) of in-/coherent sensing matrices largely coincide 
again. This nicely complements the observations made above: For very few (possibly noisy) measurements, partially coherent sensing matrices give a better reconstruction 
both in support and magnitude. But in the typical CS setting with exact recovery with respect to the \lN{2}-norm and a moderately low number of measurements, incoherent 
sensing matrices successfully prevail.

\section{Analytical Justification}\label{sec:theorems}
In this section we provide initial guarantees for the \hbox{$d$-tol}\-erant recovery of $S$-sparse signals without measurement noise through an OMP-like algorithm using 
partially coherent sensing matrices. We will utilize the notion of spread support.
\begin{df}[$d$-spread set]
  A set $B$ is $d$-spread, if
  \begin{equation}
    \forall i\neq j\in B:\enskip \abs{i-j}>d\;.
  \end{equation}
\end{df}
A signal $x$ is said to have a $d$-spread support, if $\supp{x}$ is a $d$-spread set. For the purpose of this work, sufficiently spread means a signal has a 
$(4d+1)$-spread support. 
This allows us to ignore recombinations of multiple non-zeros to a single representative during reconstruction and enables us to prove results following closely the 
initial contributions made for robust recovery via standard OMP. The line of theorems we follow is based on the exact recovery condition (ERC). Within the spread signal 
support setting only minor adjustments to the theorems are necessary in the noiseless scenario to ensure validity for \hbox{$d$-tol}\-erant recovery with partially 
coherent sensing matrices. 
Further, the OMP-like algorithm has been empirically found to perform similar to OMP in this setting with and without measurement noise. The method below is an adaptation of OMP, which is similar to Band-Excluded OMP in \cite{Fannjiang2012}, in the context of ``coherence bands''.

\subsection{Algorithm}\label{sec:sub_algorithm}
To account for the ban of recombinations in the OMP algorithm we forbid new candidates for the reconstructed support to be selected from the $2d$-closure of the already 
reconstructed support, as shown in Algorithm~1. This modification ensures that every high coherence neighborhood is met exactly once and since we will assume a $(4d+1)$ 
spread for our signals in the statements of the next section, we can guarantee not to miss any non-zero by this exclusion.
\algrenewcommand\algorithmicrequire{\textbf{Input:}}
\algrenewcommand\algorithmicensure{\textbf{Output:}}
\begin{algorithm}
  \caption{Pseudo code for d-tolerant OMP (DtOMP)}\label{alg:modified_omp}
  \begin{algorithmic}[1]
    \Require{$y\in\C^M$, $S\in\N^+$, $\Phi\in\C^{M\times N}$, $d\in\N_0$}
    \Ensure{$d$-tolerant recovery $\ti{x}\in\C^N$}
  
    \Statex $k=0$, $x^{k}=0$, $\Sigma = \setbr{}$, $r^{(0)} = y$
    \Comment Initialization

    \While{$k \leq S$ and $\abs{\Sigma} < S$}
      \State $k = k+1$
      \State $b = \abs{\Phi^*r^{(k-1)}}$
      \State $\ti{n} = \argmax_{n\notin \clos{2d}{\Sigma} } \setbr{b_n}$
      \Comment Modification\label{algline:modified_omp_selection}
      \State $\Sigma = \Sigma \cup \setbr{\ti{n}}$
      \State $x^{(k)}_{\Sigma} = \Phi^\dagger_{\Sigma} y$
      \State $r^{(k)} = y-\Phi x^{(k)}$
    \EndWhile
    \State $\ti{x} = x^{(S)}$
  \end{algorithmic}
\end{algorithm}

In the algorithm, $x^{(k)}_{\Sigma}$ and $\Phi^\dagger_{\Sigma}$ are the reconstruction in the $k$th iteration restricted to the rows in the set $\Sigma$ and the 
Moore-Penrose pseudoinverse of $\Phi$ restricted to the columns with indices in $\Sigma$, respectively.

Through the exclusion of the $2d$-closures of the already recovered support, the algorithm will select at most one candidate per high correlation region. This 
modification is negligible within the scenario of signals with $(4d+1)$-spread support as all the numerical experiments for all tested parameter sets showed. 
For signals without spread support DtOMP fails at exact recovery by design due to its exclusion feature.

\subsection{Theory}\label{sec:sub_theory}
Here we present some results closely related to established results for recovery from coherent sampling.  These, like others in the literature, are for signals with spread support only. As our experiments seem to indicate, we conjecture this condition is only an artifact of the proofs, and further study should be performed to remove  this assumption.
The given theoretical reconstruction guarantees are a close analog to the ERC based incoherent OMP theory \cite{Tropp2004,Fannjiang2012}. The presented results can be 
understood as a characterization of OMP in the noiseless scenario of signals with sufficiently spread support.

Before we formulate the theorems we need to fix the notion of a $d$-approximate pair of sets $\setbr{\Sigma,\Gamma}_d$.

\begin{df}
  \label{df:d_approximate_pair_of_sets}
  Let $\Sigma, \Gamma\subset\N^+$ be sets and a $d\in\N_0$ be given. Then we have a $d$-approximate pair of sets $\setbr{\Sigma,\Gamma}_d$, if and only if
  \begin{equation}
    \Sigma \subseteq \clos{d}{\Gamma}\;\text{and}\;\Gamma \subseteq \clos{d}{\Sigma}\;,
  \end{equation}
  that is their distance in the Hausdorff-metric is at most $d$. We will call the set of all such pairs of sets containing at least one set of cardinality $S$, 
  $\mathcal{D}^S_d$.
\end{df}

The following theorem is the essential result ensuring recovery from noiseless measurements via the \hbox{$d$-tol}\-erant recovery condition (TRC). We do not 
present a proof here, since it follows similarly to previously established results \cite{Tropp2004,Fannjiang2012}. In particular, see Theorem 1 of \cite{Fannjiang2012} 
for a more general result that tolerates noise and is in terms of arbitrary \textit{coherence bands}, rather than \hbox{$d$-tol}\-erant recovery.

\begin{thm}[$d$-tolerant recovery guarantee without measurement noise]\label{thm:dtompv2_s_rec}
  Consider \eqref{eq:linear_sensing_model_w_noise} with $e=0$ and $\abs{x}_0~=~S$.
  The \hbox{$d$-tol}\-erant reconstruction of the signal can be guaranteed via DtOMP, Algorithm~\ref{alg:modified_omp}, if:
  \begin{subequations}
    \begin{align}
      &\supp{x}\text{ is $(4d+1)$-spread}\label{eq:dtompv2_s_rec_thm_spread} \displaybreak[1]\\
      &\mu_d(\Phi) \leq \text{const.} \ll 1\label{eq:dtompv2_s_rec_thm_mu_d}\\
      &\forall \setbr{A,B}_d\in\mathcal{D}^S_d:\enskip
          A\text{ is $(4d+1)$-spread},\;\ti{T}:=A\cup B,\notag\\
          &\hspace*{4em} T:=\clos{2d}{A}\;
          \lra\; \max_{j\in T^C}\norm{\Phi_{\ti{T}}^\dagger \phi_j}_1 < 1
        \tag{TRC}\label{eq:dtompv2_s_rec_thm_ratio}
    \end{align}
  \end{subequations}
  where $\setbr{A,B}_d$, $\mathcal{D}^S_d$ are given in Definition \ref{df:d_approximate_pair_of_sets} and $\mu_d(\Phi)$ is given in \eqref{eq:def_d_coherence}.
\end{thm}

The theorem allows to guarantee the \hbox{$d$-tol}\-erant recovery of any \hbox{$S$-sparse} signal from noiseless measurements using partially coherent sensing matrices. 
The 
original theory for OMP will fail for partially coherent sensing matrices since the ERC is usually not satisfied. In addition, given that the utilized sensing matrix has 
large $d'$-coherences ($d'<d$) in every \hbox{$d$-neigh}\-borhood, the reconstruction will naturally be also close in magnitude.

Note that many naturally arising sensing matrices such as overcomplete Fourier frames satisfy the condition of the theorem for some $d$. This can be seen in 
\figref{fig:correlation_measures} for the example of $\Psi$. The TRC will hold for any sufficiently small $\mu_d(\Phi)$ since the Hausdorff distance between $T^C$ and 
$\ti{T}$ is by construction larger then $d$. We also point out that \eqref{eq:dtompv2_s_rec_thm_mu_d} is primarily a lower bound on the minimal number of measurements 
$M$. This link is established using the Welch bound applied to all possible submatrices restricted to column indices that are \hbox{$(4d+1)$-spread}.

Continuing the theoretical construction as in \cite{Tropp2004}, one can ensure the TRC by imposing conditions on the cumulative coherence. 

\begin{thm}[\ref{eq:dtompv2_s_rec_thm_ratio} guarantee]\label{thm:dtompv2_dtrc_guarantee}
  Let $\Phi\in\C^{M\times N}$. Then the TRC
  holds for all $\setbr{A,B}_d\in\mathcal{D}^S_d$, if 
  \begin{align}
    \mu^C_d(\Phi, 2S-1)+\mu^C_d(\Phi, 2S) < 1 \;,\label{eq:dtompv2_dtrc_guarantee}
  \end{align}
  where $\setbr{A,B}_d$, $\mathcal{D}^S_d$ are given in Definition \ref{df:d_approximate_pair_of_sets} and $\mu^C_d(\Phi, \cdot)$ is given in 
  \eqref{eq:def_cummulative_d_coherence}.
\end{thm}

\begin{proof}\label{proof:dtompv2_dtrc_guarantee}
  The proof follows analogously to the proof of Theorem 3.5 in \cite{Tropp2004} with only minor modifications. Namely $S$ is replaced with $2S$, since the optimal 
  support of that theorem is replaced by the union $A\cup B$ of two $S$-cardinal sets, and $\mu^C(\Phi,\cdot)$ is replaced by $\mu^C_d(\Phi,\cdot)$.
\end{proof}

\begin{concl}\label{concl:dtompv2_dtrc_guarantee_renice}
  Equation \eqref{eq:dtompv2_dtrc_guarantee} of Theorem \ref{thm:dtompv2_dtrc_guarantee} holds, with the conditions stated, if any of the  
  following inequalities is satisfied:
    \begin{align}
      S &< \frac{1}{4}\paren{\mu_d(\Phi)^{-1} + 1}      \label{eq:dtompv2_dtrc_guarantee_renice_mu_d}\\
      \mu^C_d(\Phi, 2S) &< \frac{1}{2}     \;.       \label{eq:dtompv2_dtrc_guarantee_renice_mu_c_d}
    \end{align}
\end{concl}

\begin{proof}\label{proof:dtompv2_dtrc_guarantee_renice}
  Both conditions follow from the monotonic behavior of $\mu^C_d(\Phi,\cdot)$. Equation \eqref{eq:dtompv2_dtrc_guarantee_renice_mu_d} is proved using 
  \begin{equation}
    \mu^C_d(\Phi, k) \leq \mu_d(\Phi)k\;, \label{eq:lemma_coherences_lb_mu_d}
  \end{equation}
  which holds due to the monotonic increase of $\mu^C_d(\Phi, k)$ as $k$ decreases. So we have:
    \begin{align*}
      \mu^C_d(\Phi, 2S-1) + \mu^C_d(\Phi, 2S)
      &\stackrel{\mathclap{\eqref{eq:lemma_coherences_lb_mu_d}}}{\leq}
        \paren{2S-1 + 2S}\mu_d(\Phi) < 1\label{eq:proof_dtrc_guarantee_renice_mu_d} \\
      &\iff S < \frac{1}{4}\paren{\mu_d(\Phi)^{-1} + 1}\;.
    \end{align*}

  Equation \eqref{eq:dtompv2_dtrc_guarantee_renice_mu_c_d} follows immediately from the increasing behavior of $\mu^C_d(\Phi, k)$:
    \begin{align*}
      \mu^C_d(\Phi, 2S-1) + \mu^C_d(\Phi, 2S) &\leq 2\mu^C_d(2S,\Phi) < 1 \\
      &\iff \mu^C_d(2S,\Phi) < \frac{1}{2}\;,
    \end{align*}
  completing the proof.
\end{proof}

Both results given in the corollary are stronger than the original requirement but may be easier to verify. As is the case for OMP, 
\eqref{eq:dtompv2_dtrc_guarantee_renice_mu_d} is stronger than \eqref{eq:dtompv2_dtrc_guarantee_renice_mu_c_d}.  As a concrete example of a matrix that satisfies the conditions of the corollary, one could consider $F\_consecBegin$; when $N=1024$, $S=1$, and $M=64$, the conditions hold for $13\leq d \leq 23$. When $N=512$, $S=2$, and $M=64$, the conditions hold for $15\leq d \leq 20$.  Clearly these are not optimal conditions, but do provide a heuristic that holds for practical sensing matrices in certain parameter regimes. Moreover, we see tolerant support recovery for much broader ranges in the experiments.

\section{Conclusion and future directions}\label{sec:conclusion_and_future_directions}
We considered \hbox{$d$-tol}\-erant recovery and showed that in the low and noisy measurement regime, coherence in the sensing matrix is actually beneficial -- despite 
just the 
opposite in the classical recovery setting. We have taken first steps towards developing a framework and building a theoretical foundation for 
\hbox{$d$-tol}\-erant recovery. An empirical characterization of OMP for the purpose of \hbox{$d$-tol}\-erant recovery has been provided. It was backed for signals with 
sufficiently spread support theoretically through the interim modified OMP, termed DtOMP, which was found to be empirically the same as OMP in this setting. 
A comparison with simpler downsampling alternatives and the convex SPGL1 BPDN algorithm underlined our findings and showed best performance, both with respect to 
tolerant support and tolerant magnitude recovery, for partially coherent sensing matrices when paired with DtOMP. The modifications necessary for the ERC based OMP 
reconstruction guarantees were minimal. We introduced a modified version of the ERC, called TRC, with which we were able to prove \hbox{$d$-tol}\-erant recovery of 
arbitrary $S$-sparse signals with $(4d+1)$ spread support from noiseless measurements using partially coherent sensing matrices. For noisy recovery the classic proofs 
can not be easily extended.

Some future directions include:
(i) developing new prove strategies to proof recovery guarantees for the noisy measurement setting;
(ii) deriving theoretical guarantees of the \hbox{$d$-tol}\-erant reconstruction for signals without a spread support; 
(iii) analyzing the coherent sensing paradigm for other algorithms in order to improve the reconstruction performance; 
(iv) deriving a characterization of the phase transition of \hbox{$d$-tol}\-erant algorithms, to enable clear assertions whether partially  
coherent or incoherent sensing should be employed given the problem dimensions;
(v) further investigating how the smallest "optimal" value of $d$ relates to the problem dimensions and the coherence levels in the matrix in general;
(vi) the development of partially coherent sensing matrices specifically designed for particular applications.
(vii) a study utilizing a variant of the restricted isometry property may be illuminating in the context of tolerant recovery, see 
\cite{Fu2014} for some initial considerations in this direction.

\section*{Acknowledgments}
We would like to thank Omer Bar-Ilan, whose Master thesis inspired this work, the Faculty of Mathematics and Informatics of TU Bergakademie Freiberg for providing 
computational resources.

The work of T. Birnbaum was supported in part by the INTERFERE ERC Consolidator Grant. The work of Y. C. Eldar was supported in part by the European 
Union’s Horizon 2020 Research and Innovation Program through the ERC-BNYQ Project, and in part by the Israel Science Foundation under Grant 335/14. The work of D. 
Needell was supported by NSF CAREER grant $\sharp$1348721 and the Alfred P. Sloan Foundation.

\bibliographystyle{unsrturl}
\bibliography{CohSenPaper}{}

\end{document}